%% file: 0_main.tex
\documentclass[10pt,journal,compsoc]{IEEEtran}
\ifCLASSOPTIONcompsoc
  \usepackage[nocompress]{cite}
\else
  \usepackage{cite}
\fi
\ifCLASSINFOpdf
\else
\fi
\usepackage{amsmath, amsthm, amssymb}
\newtheorem{theorem}{Theorem}[section]

\newtheorem{definition}{Definition}[section]
\newtheorem{lemma}[theorem]{Lemma}
\usepackage{multirow}
\usepackage{booktabs}
\usepackage{threeparttable}
\usepackage{graphicx}
\usepackage{makecell}
\usepackage{bbding}
\usepackage{hyperref}
\usepackage{subfigure}

\usepackage{xspace}
\usepackage{xcolor}

\usepackage{algorithm}
\usepackage{algpseudocode}
\begin{document}
\title{Hypergraph Isomorphism Computation}

\author{Yifan Feng, Jiashu Han,
        Shihui Ying,~\IEEEmembership{Member,~IEEE} and Yue Gao,~\IEEEmembership{Senior Member,~IEEE}% <-this % stops a space
\IEEEcompsocitemizethanks{
\IEEEcompsocthanksitem Yifan Feng, Jiashu Han, and Yue Gao are with the School of Software, BNRist, THUIBCS, BLBCI, Tsinghua University, Beijing 100084, China. \protect\\
E-mail: evanfeng97@gmail.com; joshuahan1228@gmail.com; gaoyue@tsinghua.edu.cn; 
\IEEEcompsocthanksitem Shihui Ying is with the Department of Mathematics, School of Science, Shanghai University, Shanghai 200444, China. \protect\\
E-mail: shying@shu.edu.cn;
}
% \IEEEcompsocitemizethanks{\IEEEcompsocthanksitem M. Shell was with the Department
% of Electrical and Computer Engineering, Georgia Institute of Technology, Atlanta,
% GA, 30332.\protect\\
% % note need leading \protect in front of \\ to get a newline within \thanks as
% % \\ is fragile and will error, could use \hfil\break instead.
% E-mail: see http://www.michaelshell.org/contact.html
% \IEEEcompsocthanksitem J. Doe and J. Doe are with Anonymous University.}% <-this % stops an unwanted space
% \thanks{Manuscript received April 19, 2005; revised August 26, 2015.}
}


% \markboth{Journal of \LaTeX\ Class Files,~Vol.~14, No.~8, August~2015}%
\markboth{IEEE TRANSACTIONS ON PATTERN ANALYSIS AND MACHINE INTELLIGENCE}
{Shell \MakeLowercase{\textit{et al.}}: Bare Demo of IEEEtran.cls for Computer Society Journals}



\IEEEtitleabstractindextext{%
\begin{abstract}
The isomorphism problem is a fundamental problem in network analysis, which involves capturing both low-order and high-order structural information. In terms of extracting low-order structural information, graph isomorphism algorithms analyze the structural equivalence to reduce the solver space dimension, which demonstrates its power in many applications, such as protein design, chemical pathways, and community detection. 
%However, in the real world, high-order correlation (hypergraph) is very common, and some of them cannot be solved well by existing graph isomorphism algorithms.
For the more commonly occurring high-order relationships in real-life scenarios, the problem of hypergraph isomorphism, which effectively captures these high-order structural relationships, cannot be straightforwardly addressed using graph isomorphism methods.
Besides, the existing hypergraph kernel methods may suffer from high memory consumption or inaccurate sub-structure identification, thus yielding sub-optimal performance. 
%In this paper, to address the abovementioned problems, we first generalize the Weisfiler-Lehman test algorithm from graphs to hypergraphs and propose the hypergraph Weisfiler-Lehman test algorithm for the hypergraph isomorphism test problem.
In this paper, to address the abovementioned problems, we first propose the hypergraph Weisfiler-Lehman test algorithm for the hypergraph isomorphism test problem by generalizing the Weisfiler-Lehman test algorithm from graphs to hypergraphs. 
%It is achieved by devising a two-stage neighbor relation to model the complex high-order correlation in hypergraphs.%
%Secondly, based on the presented algorithm, we propose a general hypergraph Weisfieler-Lehman kernel framework, which can transform the comparison of hypergraph structures into the inner product of feature vectors. We further implement two instances: Hypergraph Weisfeiler-Lehamn Subtree Kernel and Hypergraph Weifeiler-Lehman Hyperedge Kernel.
Secondly, based on the presented algorithm, we propose a general hypergraph Weisfieler-Lehman kernel framework %that can transform the comparison of hypergraph structures into the inner product of feature vectors
and implement two instances, which are Hypergraph Weisfeiler-Lehamn Subtree Kernel (Hypergraph WL Subtree Kernel) and Hypergraph Weisfeiler-Lehamn Hyperedge Kernel (Hypergraph WL Hyperedge Kernel). 
%The Hypergraph WL Subtree Kernel counts different types of rooted subtrees and generates the final feature vector for a given hypergraph. It compares the number of different types of rooted subtrees to measure the distance of given hypergraphs.
The Hypergraph WL Subtree Kernel counts different types of rooted subtrees and generates the final feature vector for a given hypergraph by comparing the number of different types of rooted subtrees. % to measure the distance of given hypergraphs.
The Hypergraph WL Hyperedge Kernel is developed to process hypergraphs with more degrees of hyperedges, which counts the vertex labels that are connected by each hyperedge to generate the feature vector. % for a given hypergraph.% 
%Thirdly, we provide mathematical proof that the proposed Hypergraph WL Subtree Kernel can degenerate into the typical Graph Weisfeiler-Lehman Subtree Kernel confronting the low-order graph structure.
Mathematically, we prove the proposed Hypergraph WL Subtree Kernel can degenerate into the typical Graph Weisfeiler-Lehman Subtree Kernel when dealing with low-order graph structures.
In order to fulfill our research objectives, a comprehensive set of experiments was meticulously designed %We conduct extensive experiments on $19$ synthetic and real-world datasets
, including seven graph classification datasets and $12$ hypergraph classification datasets. Results on graph classification datasets indicate that the Hypergraph WL Subtree Kernel can achieve the same performance compared with the classical Graph Weisfeiler-Lehman Subtree Kernel. Results on hypergraph classification datasets show significant improvements compared to other typical kernel-based methods, which demonstrates the effectiveness of the proposed methods. %We also evaluate the runtime of the proposed methods compared with the compared second-best method. Our method can run more than $80$ times faster than the compared second-best method when confronting more complex hypergraph structure, which shows the large potential of the proposed methods in real-world applications.
In our evaluation, we found that our proposed methods outperform the second-best method in terms of runtime, running over 80 times faster when handling complex hypergraph structures. This significant speed advantage highlights the great potential of our methods in real-world applications.

\end{abstract}

}

\maketitle

\IEEEdisplaynontitleabstractindextext

% \fi
\IEEEpeerreviewmaketitle

\input{1_intro}

\input{2_related}

\input{3_preliminary}
\input{4_method}
\input{5_exp}

\input{6_conclusion}

\input{7_app}
\input{8_bib}

\input{9_bio}
\end{document}

%% file: 1_intro.tex
\IEEEraisesectionheading{\section{Introduction}\label{sec:introduction}}

% 图数据广泛存在应用广泛。更多的数据是高阶关联的。图无法建模高阶关联数据，因此超图慢慢火起来。超图有很多应用

% 这些应用依赖对结构有一定的辨别能力。同构测试是这些应用的基础。图有一些结构辨别的方法。然而超图这方面很少

% 现有超图辨别方法怎么做，有什么缺点。

% 这篇文章我们针对超图同构问题，提出了一套算法。这套算法怎么做。

% 同时基于这套算法，我们进一步提出了两种核方法。怎么做的。在很多数据集上做了大量的实验。图数据集上的结果。超图数据集上的结果。总结一下我们的贡献。

% 贡献1：提出了超图WL算法，直接能用于超图上的结构辨别
% 贡献2：基于算法，提出了两个核函数，并提供了理论的证明，在图上可以退化成图WL，在超图上更好
% 贡献3：做了大量的数据集和实验，证明了我们方法快速且有效 

%围绕现实世界中的关联建模的重要性

%先讲同构这个问题的重要性
%传统方法：把图方法归类到传统方法中
% 虽然

\IEEEPARstart{N}{etwork}, as a typical irregular modeling tool, has shown its advances in many applications like social networks\cite{graph_social_network}, brain networks\cite{graph_brain}, collaborative networks\cite{graph_collabrative_network}, knowledge networks\cite{graph_knowledge_network}, chemical pathways\cite{graph_pathways}, and protein structures\cite{graph_protein}. However, real-world irregular data often comprise numerous high-order correlations that a simple graph structure cannot adequately represent. When attempting to model high-order correlations, such as the group relations in social media \cite{hg_social_media, hgnn, hgnnp} or the co-author relations in academic papers \cite{dhgnn, hg_coauthor, hg_encoder}, the intricate correlations among vertices become ambiguous, as illustrated in Figure \ref{fig:intro}. The hyperedge in hypergraphs can link more than two vertices, which endows the hypergraph with the capability to model beyond pair-wise correlation compared with the simple graph. To analyze and understand the hypergraph structure, one needs a tool to measure the similarity between different hypergraphs or sub-structures of a hypergraph, also known as the ``Isomorphism Test''.

%为了解决问题，超图被用于刻画。。。
%然而超图作为新鲜事物，存在什么什么问题 （同构性判定） 
%超图的同构判定有什么结果，介绍前期工作
%受到图同构判断的启发。。。

% 分两类，超图转图 然后再说图铜鼓   直接超图同构，再bulabula
% 将 图同构作为超图同构的一条路径
% 有很多图同构的方法，但是都无法应用用于超图。

%While graph isomorphism algorithms excel at extracting low-order structural information, they are limited in their ability to handle relationships beyond pair-to-pair connections. There exist many graph similarity measures based on graph isomorphisms like the Weisfeiler-Lehman test \cite{graph_wl} of graph isomorphism and graph kernels\cite{graph_wl_subtree, graph_wl_all, graphlet}.

Graph isomorphism algorithms, such as the Weisfeiler-Lehman test \cite{graph_wl} for graph isomorphism and graph kernels \cite{graph_wl_subtree, graph_wl_all, graphlet}, excel at extracting low-order structural information. However, they face limitations when dealing with relationships that extend beyond pair-to-pair connections. In the typical Weisfeiler-Lehman test, each vertex gathers its neighbor vertices' labels to generate the compressed vertex labels. Each compressed vertex label corresponds to a unique subtree structure. Nevertheless, the algorithm can only output whether the two graphs are identity. Furthermore, based on the algorithm, the graph kernel methods are proposed to measure the similarity of two graphs with a continuous value. The graph kernel methods design a mapping from the graph structure to the vector in the feature space and utilize the feature vector's inner product to measure the graphs' similarity. However, these graph-based similarity measurement methods, which rely on vertex-to-vertex label propagation, are not suitable for handling the complexity of such high-dimensional structures, as mentioned earlier in the case of hypergraphs.

\begin{figure}
    \centering
    \includegraphics[width=0.45\textwidth]{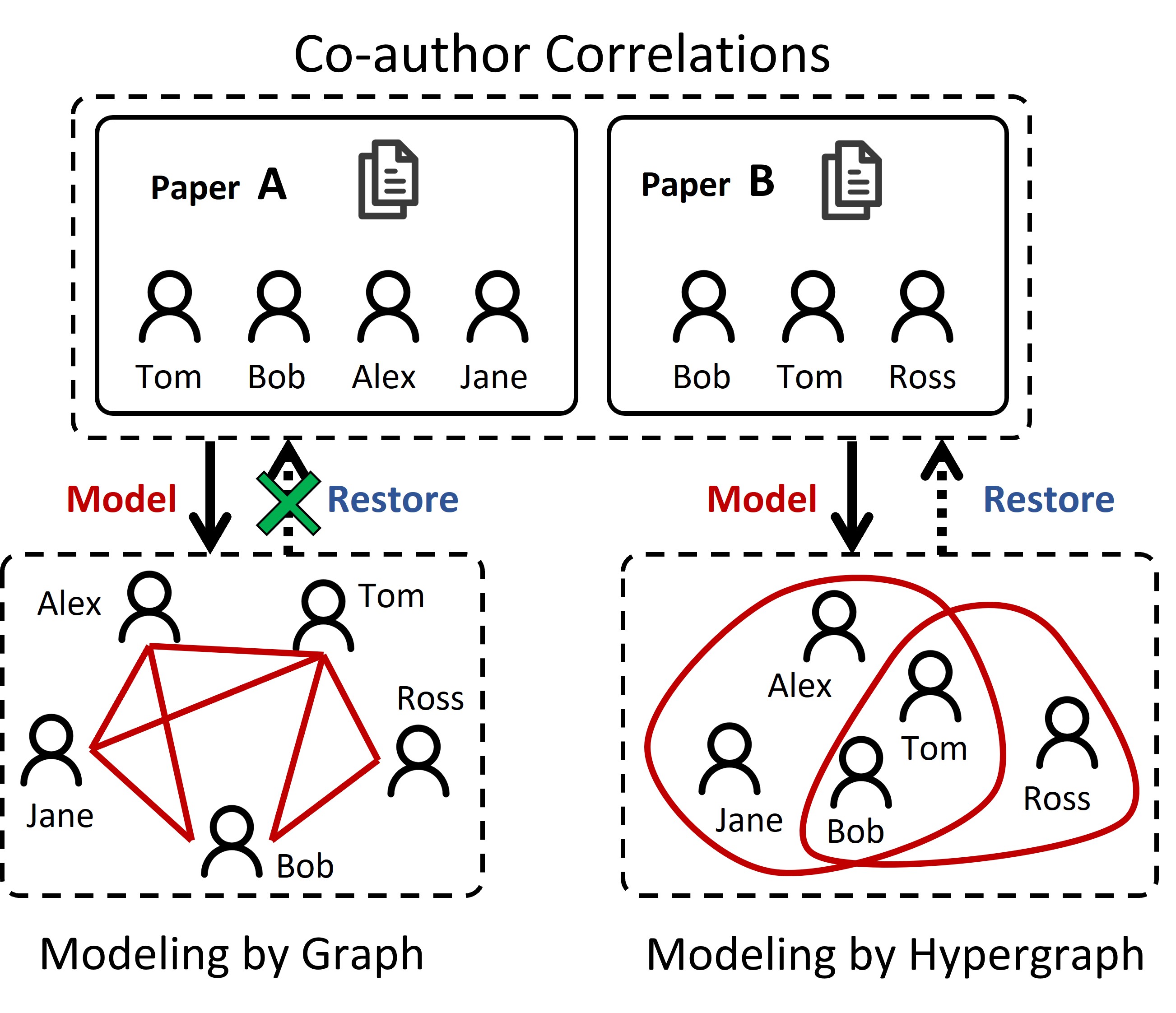}
    \caption{An illustration of modeling the co-author correlations by graph or hypergraph. The line in red color denotes the ``co-author'' relation. Compared with the graph, the hypergraph can accurately model the high-order correlations and restore the original correlations.}
    \label{fig:intro}
\end{figure}

With the emergence of hypergraph structures, researchers have also proposed various strategies \cite{hg_iso, hg_line, hg_root} to address the problem of high-order structure similarity measurement. However, most of these methods fundamentally rely on transforming hypergraphs into graph structures. These indirect hypergraph isomorphism algorithms, to some extent, sacrifice the complex structural information inherent in hypergraphs, leading to misalignments of criteria and cumbersome computational processes.
%Recently, researchers also proposed some strategies\cite{hg_iso, hg_line, hg_root} for the hypergraph similarity measure problem. Wachman et al.
\cite{hg_root} develop a sampling-based method that draws the vertex-hyperedge sequence from the original hypergraph. However, it will bring information loss and fail in many common cases (the pyramid-like hypergraphs in Figure \ref{fig:pyramid}).
By employing an indirect approach, Bai et al. \cite{hg_line} design a transformation-based method, which transforms the hypergraph structure into a low-order graph structure with three steps. In this way, the scale of the generated graph will be extremely large and can not be employed in practice (causing the out-of-memory error in the dataset containing more than 3000 hypergraphs in Table \ref{tab:hypergraph_real}).

To address those problems, we first extend the Weisfier-Lehman test algorithm \cite{graph_wl} from graphs to hypergraphs. The definition of the neighbor relation is the obstacle to applying the Weisfiler-Lehman test in hypergraphs. Thus, as for the complex high-order correlations, we decompose the neighbor relation into two sub-neighbor relations: vertex's hyperedge neighbor and hyperedge's vertex neighbor. Then, based on the two neighbor relations, we devise a two-stage hypergraph Weisfiler-Lehman test algorithm, which can be directly applied to hypergraphs. Second, based on the hypergraph Weisfeiler-Lehman test algorithm, we proposed a general hypergraph Weisfeiler-Lehman kernel, which maps the hypergraph structure into a vector in the feature space. With the hypergraph Weisfeiler-Lehman kernel, the distance between two hypergraphs can be computed by the inner product of the two corresponding feature vectors. Third, we implement two instances of the general hypergraph Weisfeiler-Lehman kernel: hypergraph Weisfeiler-Lehman subtree kernel and hypergraph Weisfeiler-Lehman hyperedge kernel. The hypergraph Weisfeiler-Lehman subtree kernel directly counts different types of subtree structures, and the hypergraph Weisfeiler-Lehman hyperedge kernel counts different  hyperedges. Besides, to deeply exploit the relation between the graph Weisfeiler-Lehman subtree kernel and the proposed hypergraph Weisfeiler-Lehman subtree kernel, we have proven that the hypergraph Weisfeiler-Lehman subtree kernel can be reduced to the graph Weisfeiler-Lehman subtree kernel when processing the graph structures from a mathematical perspective. To verify the effectiveness of the proposed hypergraph Weisfeiler-Lehman kernels, we conduct experiments on $19$ graph/hypergraph classification datasets, including two synthetic graph datasets, five real-world graph datasets, four synthetic hypergraph datasets, and eight real-world hypergraph datasets. Experimental results on graph datasets demonstrate that the proposed hypergraph Weisfeiler-Lehman subtree kernel can achieve the same performance as the graph Weisfeiler-Lehman subtree kernel (proven by Theorem \ref{theorem:equal}). Experimental results on hypergraph datasets show a superior performance increase compared with other kernel-based methods. We also conduct experiments to compare the runtime with other hypergraph methods, which is important in practice. Results demonstrate that our methods run faster and can achieve better performance. The main contributions of this paper are summarized as follows:

\begin{itemize}
    \item We extend the Weisfeiler-Lehman test from graphs to hypergraphs and propose the hypergraph Weisfeiler-Lehman test algorithm for the hypergraph isomorphism test problem. % 提出框架，他是对现有图方法的一个拓展
    \item We propose a general hypergraph Weisfeiler-Lehman kernel framework for hypergraph or sub-structure hypergraph similarity measure and implement two instances: hypergraph Weisfeiler-Lehman subtree kernel and hypergraph Weisfeiler-Lehman hyperedge kernel.
    \item We mathematically prove that the proposed hypergraph Weisfeiler-Lehman subtree kernel can degenerate into graph Weisfeiler-Lehman subtree kernel and achieve the same performance confronting graph structures. 
    \item Extensive experiments on $19$ graph/hypergraph classification datasets demonstrate the effectiveness of the proposed methods.
\end{itemize}

%% file: 2_related.tex
\section{Related Work}
In this section, we will introduce some related graph isomorphism methods and hypergraph isomorphism methods. 
% The graph isomorphism problem has been a prominent topic in graph theory for almost four decades. Prior to the 21st century, tackling the graph isomorphism problem from a traditional perspective was a lively field, resulting in the development of various methods for testing the isomorphism between two graphs, such as the WL-Test and Subgraph Patterns. 
% With the rapid advancement of Kernel methods in the field of machine learning, many graph kernel-based methods emerged based on traditional graph isomorphism algorithms, offering a fresh perspective for extracting structural features from graphs. 
\subsection{Graph Isomorphism Methods}
%Those methods enable the extraction of meaningful feature information from the structure of graphs.

The Weisfeiler-Leman Algorithm, also known as the WL-test, was initially proposed in 1968\cite{graph_wl}. Weisfeiler and Leman made a pioneering contribution by applying the Color Refinement algorithm to the field of graph isomorphism. The fundamental concept behind this algorithm involves labeling the vertices of a graph based on their iterated degree sequence. In 1992, Babai and Mathon\cite{graph_k_wl} extended the Color Refinement process, introducing the k-dimensional Weisfeiler-Leman algorithm (k-WL). In contrast to the two-dimensional version, Babai et al. \cite{graph_k_wl} employed color tuples of vertices instead of single vertex coloring.

The Group-Theoretic Graph Isomorphism Machinery, known as Subgraph Patterns, was proposed by Luks et al.\cite{group_theory} in 1982 to test the isomorphism of graphs with bounded degrees. Luks et al. \cite{group_theory} pioneered a broader problem-solving strategy for graph isomorphism by devising a recursive mechanism that leverages the structure of permutation groups to encode graph structures. This algorithm forms the cornerstone of the group-theoretic graph isomorphism machinery, providing a solid foundation for further developments in the field. One of the most representative methods in this domain is the Graphlet method proposed by Prˇzulj et al. \cite{motif_graphlet}. Graphlets are subgraph patterns, with each graphlet representing an instance of an isomorphism type. Kondor demonstrated that vectorizing the statistical frequency of all kernel occurrences can effectively embed the structural features of graphs into a feature space.

The concept of Graph Kernel was first introduced in \cite{graph_iso_np_hard, graph_kernel_two} in 2003. Jan Ramon and Thomas Gartner et al. \cite{efficiency_graph_kernel} pioneered the Subtree Graph Kernel in 2003, which marked the emergence of subtree-based isomorphism algorithms. Jan Ramon rigorously derived and computed to demonstrate that encoding subtree structures can significantly enhance expressive power while incurring a minimal computational cost. In the realm of Graph Kernel, there are three main directions of research. Firstly, the Shortest-path graph Kernel, which is the earliest proposed graph kernel isomorphism algorithm, randomly extracts subgraph structural information through random walks. During the same period, the Weisfeiler-Lehman Subtree Kernel achieved remarkable expressive power improvement by employing the WL-test with minimal computational cost. Thirdly, the Graphlet Kernel algorithm proposed a solution to the incompleteness problem of the WL-test. It statistically counts the number of kernels in the graph through random sampling, thereby encoding the structural information of the graph data. Notably, Shervashidze's groundbreaking research in 2009 \cite{graphlet} significantly improved this method by approximating the occurrence probabilities of Graphlets through random sampling, enhancing computational efficiency. Although this approach reduced memory consumption in practical applications, it compromised the reliability and stability of the algorithm due to the random sampling technique employed.

\subsection{Hypergraph Isomorphism Methods}
The development of hypergraph-based isomorphism testing has seen vigorous growth in recent years, coinciding with the rise of hypergraph structures. As early as 2007, Wachman et al. \cite{hg_root} experimented with isomorphism algorithms on hypergraphs. Due to the inherent unreliability of random walk algorithms, the kernel function suffered from high complexity and unstable learning structures. In 2008, Laszl'o et al.\cite{hg_low_rank} investigated the computational aspects of hypergraph isomorphism and proposed algorithms that can handle hypergraphs with low-rank structures efficiently within a reasonable computational timeframe. Arvind et al. \cite{CHI} presented a fixed parameter tractable algorithm for Colored Hypergraph Isomorphism. Arvind claimed their algorithm could be seen as a generalization of Luks’s result for Hypergraph Isomorphism. In 2014, Bai et al.\cite{hg_line} tackled this issue by transforming the hypergraph into a bipartite graph, enabling the application of the WL-test to hypergraph structures. However, this transformation significantly increased the size of the hypergraph, making the algorithm excessively complex. The low computational efficiency prevented the direct application of the algorithm to hypergraph data. As stated by the authors in its Conclusion \cite{hg_line}, "Our future work is to develop a new higher-order WL isomorphism test algorithm that can be directly performed on a hypergraph." The field of hypergraph isomorphism requires a WL isomorphism algorithm that can directly operate on hypergraphs, which are the methods proposed in our paper.

%% file: 3_preliminary.tex
% \section{Notations}

\section{Preliminary}
In this section, we first introduce the problem background of the graph/hypergraph isomorphism. Then, we briefly review the most related graph Weisfeiler-Lehman kernel. The detailed descriptions of notations are in Table \ref{tab:nots}.

\begin{table}%[!htbp]
    \centering
    \label{tab:nots}
    \caption{Notations and associated descriptions in this paper.}
    \begin{tabular}{cc}
    \toprule
    Notations & Descriptions \\
    \midrule
    $G$ & A graph \\
    $V$ & Vertex set of a graph \\
    $E$ & Edge set of a graph \\
    $\ell(v)$ & Label of vertex $v$ \\
    $M_\cdot(v)$ & Multiset with respect to vertex $v$ \\
    $\mathcal{N}(v)$ & Set of neighbor vertices of vertex $v$ \\
    \midrule
    $\mathcal{G}$ & A hypergraph \\
    $\mathcal{V}$ & Vertex set of a hypergraph \\
    $\mathcal{E}$ & Hyperedge set of a hypergraph \\
    $\mathbf{H}$ & Hypergraph incidence matrix \\
    $\ell^v(v)$ & Label of vertex $v$ \\
    $\ell^e(v)$ & Label of hyperedge $e$ \\
    $\mathcal{N}_e(v)$ & Set of vertex's hyperedge neighbors of vertex $v$ \\
    $\mathcal{N}_v(e)$ & Set of hyperedge's vertex neighbors of hyperedge $e$ \\
    $M^e_\cdot(v)$ & Multiset of hyperedges with respect to vertex $v$ \\
    $M^v_\cdot(e)$ & Multiset of vertices with respect to hyperedge $e$ \\
    \midrule
    $h$ & Number of iterations \\
    $N$ & Number of hypergraphs \\
    $\phi(\cdot)$ & Feature map \\
    $c(X, \sigma)$ & Number of sub-structure $\sigma$ in structure $X$ \\
    $k_\cdot(x, x')$ & Kernel function with respect to $x$ and $x'$ \\
    \bottomrule
    \end{tabular}
\end{table}

\subsection{Graph/Hypergraph Isomorphism}

\textbf{Graph Isomorphism:} 
Given two graphs $G = \{V, E\}$ and $G' = \{V', E'\}$, the target of the graph isomorphism test (denoted by $G \cong G'$) is to find whether a bijective mapping $g := V \rightarrow V'$ exists. The mapping $g$ is called the isomorphism function, such that 
\begin{equation}
\nonumber
    (v_i, v_j) \in E \iff (g(v_i), g(v_j)) \in E' .
\end{equation}

However, the complete graph isomorphism test solution has been proven to be the NP-hard problem \cite{graph_iso_np_hard}, thus making it practically infeasible. The Weisfeiler-Lehman test \cite{graph_wl}, also known as ``naive vertex refinement'', is a typical approximate solution of the graph isomorphism test, as described in Algorithm \ref{alg:g_wl}. It starts with two input graphs $G$ and $G'$ associated with labeled vertices. Each vertice gathers its neighbor vertices' labels in each iteration to build a subtree string, which is used to relabel the vertex. For each iteration, if the statics of different types of vertex labels of the two graphs are different, the algorithm is then determined, and the two graphs are not isomorphism graphs. Otherwise, the two graphs may be isomorphisms to each other.

\begin{algorithm}[!htbp]
	\renewcommand{\algorithmicrequire}{\textbf{Input:}}
	\renewcommand{\algorithmicensure}{\textbf{Output:}}
	\caption{Weisfeiler-Lehman test of graph isomorphism.}
	\label{alg:g_wl}
	\begin{algorithmic}
		\Require Graph ${G}=\{{V}, {E} \}$ and ${G'} = \{ {V}', {E}' \}$, vertex $v \in {V}$, and $v' \in {V}'$, vertex label map: $\ell := v/v' \rightarrow c$, hash function $f := s \rightarrow c$, $\cdot / \cdot$ operation denotes the `` or ''.
            \State \emph{1. Multiset-label determination.}
            \begin{itemize}
                \item For $i=0$, set $M_i(v) := l_0(v)=\ell(v)$.
                \item For $i>0$, assign a multiset-label $M_i(v) = \{ l_{i-1}(u) \mid u \in \mathcal{N}(v) \}$ to each node $v$ in $G$ and $G'$.
            \end{itemize}
            \State \emph{2. Sorting each multiset.}
            \begin{itemize}
                \item Sort elements in $M_i(v)$ and transfer to string $s_i(v)$. 
                \item Add $l_{i-1}(v)$ as prefix to $s_i(v)$.
            \end{itemize}
            \State \emph{3. Label compression.}
            \begin{itemize}
                \item Sort all of the string $s_i(v)$ for all $v$ in $G$ and $G'$.
                \item Map each string $s_i(v)$ to a new compressed label.
            \end{itemize}
            \State \emph{4. Relabeling.}
            \begin{itemize}
                \item Set $l_i(v) := f(s_i(v))$ for all vertices in $G$ and $G'$.
            \end{itemize}
	\end{algorithmic}  
\end{algorithm}

\textbf{Hypergraph Isomorphism:} Similarly, given two hypergraphs $\mathcal{G} = \{ \mathcal{V}, \mathcal{E} \}$ and $\mathcal{G}' = \{ \mathcal{V}', \mathcal{E}' \}$, the target of the hypergraph isomorphism test  (denoted by $\mathcal{G} \cong \mathcal{G}'$) is to find whether a bijective mapping $g := \mathcal{V} \rightarrow \mathcal{V}'$ exists. The mapping $g$ is called the isomorphism function, such that 
\begin{equation}
\nonumber
    (v_1, v_2, \cdots, v_m) \in \mathcal{E} \iff (g(v_1), g(v_2), \cdots, g(v_m)) \in \mathcal{E}' .
\end{equation}

\subsection{Graph Weisfeiler-Lehman Kernel}
Based on the Weisfeiler-Lehman test of graph isomorphism, the graph Weisfeiler-Lehman subtree kernel \cite{graph_wl_subtree} and graph Weisfeiler-Lehman edge kernel \cite{graph_wl_all} are proposed, respectively. The core motivation behind the two kernel methods is mapping the uncomparable graph structure into a feature vector in feature space, where the distance between two graphs can be calculated by the inner product of the two feature vectors as follows
\begin{equation}
    k(G, G') = \phi(G)^\top \phi(G') .
\end{equation}

\subsubsection{Graph Weisfeiler-Lehman Subtree Kernel}
Due to the vertex label in each iteration corresponding to a unique subtree structure, the simple idea behind the graph Weisfeiler-Lehman subtree kernel is to count the number of different types of labels in the graphs. Thus, the feature function $\phi$ can be defined as
\begin{equation}
\nonumber
    \phi(G) = (c(G, \sigma_{00}), \cdots, c(G, \sigma_{0|\Sigma_0|}), \cdots, c(G, \sigma_{h|\Sigma_h|})) ,
\end{equation}
where $c(G, \sigma)$ is the counting function. $\sigma_{ji}$ is the corresponding label of the $j$-th subtree in $i$-th iteration, and $\Sigma_i$ is the number of different subtree types in $i$-th iteration, with a maximum $h$ numbner of iterations.

\subsubsection{Graph Weisfeiler-Lehman Edge Kernel}
Furthermore, the graph Weisfiler-Lehman edge kernel counts matching pairs of edges with identically labeled endpoints, which can be formulated as
\begin{equation}
\nonumber
    \phi(G) = (c(G, \sigma_{0k_0} \rightarrow \sigma_{0k_0'}), \cdots, c(G, \sigma_{hk_h} \rightarrow \sigma_{hk'_h})) ,
\end{equation}
where $k_i, k'_i \in [1, |\Sigma_i|]$. $\sigma_{ik_i} \rightarrow \sigma_{ik_i'}$ denotes an edge with two labeled vertices: $\sigma_{ik_i}$ and $\sigma_{ik'_i}$. 

%% file: 4_method.tex
\section{Methods}
In this section, we first introduce the Weisfeiler-Lehman test generalized for hypergraph, named hypergraph Weisfeiler-Lehman test. Then, we define the hypergraph Weifeiler-Lehman sequence and general hypergraph kernel based on them. In the following, we present two instances of this kernel: the hypergraph Weisfeiler-Lehman subtree kernel and the hypergraph Weisfeiler-Lehman hyperedge kernel. Finally, we compare our method to Graph Weisfeiler-Lehman Kernel from a mathematical perspective. % and the General Hypergraph Neural Network (HGNN$^+$) Framework.

\begin{figure}%[!htbp]
    \centering
    \includegraphics[width=0.42\textwidth]{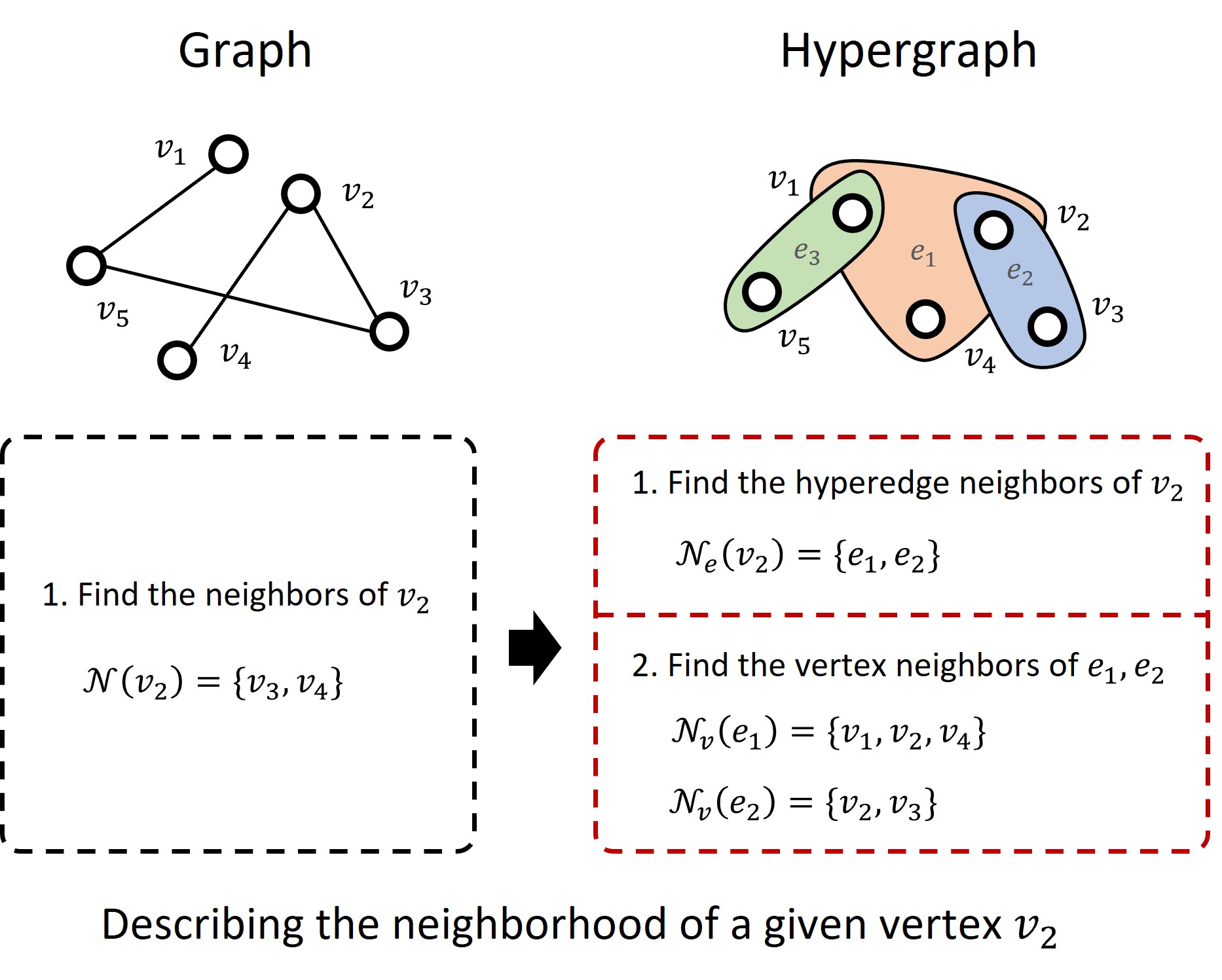}
    \caption{Illustration of generalizing the definition of neighbor from graphs to hypergraphs.}
    \label{fig:neighbor}
\end{figure}

\subsection{Hypergraph Weisfeiler-Lehman Test}
Based on the concept of the Weisfeiler-Lehman test \cite{graph_wl} of isomorphism on graphs, in this subsection, we generalize it into hypergraphs as presented in Algorithm \ref{alg:hg_wl}.

\begin{algorithm}%[!htbp]
	%\textsl{}\setstretch{1.8}
	\renewcommand{\algorithmicrequire}{\textbf{Input:}}
	\renewcommand{\algorithmicensure}{\textbf{Output:}}
	\caption{The hypergraph Weisfeiler-Lehman test of hypergraph isomorphism algorithm.}
	\label{alg:hg_wl}
	\begin{algorithmic}[1]
		\Require Hypergraph $\mathcal{G}=\{\mathcal{V}, \mathcal{E} \}$ and $\mathcal{G'} = \{ \mathcal{V}', \mathcal{E}' \}$, vertex $v \in \mathcal{V}$, and $v' \in \mathcal{V}'$, hyperedge $e \in \mathcal{E}$ and $e' \in \mathcal{E}'$, vertex label map: $\ell^v := v/v' \rightarrow c$, hyperedge label map $\ell^e := e/e' \rightarrow c$, number of iterations $h$, multiset sorting function $sort(\cdot)$, concatenating elements in multiset and stringify function: $str(\cdot)$, string compression function $f(\cdot)$, concatenating function to concate a new stringify function with the old one: $strcat(\cdot, \cdot) $, $\cdot / \cdot$ operation denotes the `` or ''. 
            \State $l_0^v(v) \gets \ell^v(v)$, $v \in \mathcal{V}$ \Comment Initialize $\mathcal{G}$ vertex labels.
            \State $l_0^v(v') \gets \ell^v(v')$, $v' \in \mathcal{V}'$ \Comment Initialize $\mathcal{G}'$ vertex labels.
            \State $l_0^e(e) \gets \ell^e(e)$, $e \in \mathcal{E}$ \Comment Initialize $\mathcal{G}$ hyperedge labels.
            \State $l_0^e(e') \gets \ell^e(e')$, $e' \in \mathcal{E}'$ \Comment Initialize $\mathcal{G}'$ hyperedge labels.
            \State $i \gets 0$
            \While{$ i < h$ and $\{l_i^v(v) \mid v \in \mathcal{V}\} = \{l_i^v(v') \mid v' \in \mathcal{V}'\}$}
		\State $i \gets i + 1$
		\State \emph{// Each hyperedge gathers its vertex neighbor labels.}
            \For{$e/e' \in \mathcal{E}/\mathcal{E}'$}
            \State \emph{// 1. Gathering hyperedge's vertex neighbor labels into multiset.}
            \State $M^v_i(e) \gets \{l^v_{i-1}(v) \mid v \in \mathcal{N}_v(e) \}$ 
            \State $M^v_i(e') \gets \{l^v_{i-1}(v') \mid v' \in \mathcal{N}_v(e') \}$ 
            \State \emph{// 2. Sorting multiset and stringify.}
            \State $S_e / S_{e'} \gets str(sort(M^v_i(e))) / str(sort(M^v_i(e')))$
            \State \emph{// 3. Hyperedge label compression and Relabeling.}
            \State $S_e / S_{e'} \gets strcat(l_{i-1}^e(e), S_e) / strcat(l_{i-1}^e(e'), S_{e'})$
            \State $l_{i}^e(e) / l_{i}^e(e') \gets f(S_e) / f(S_{e'})$
            \EndFor
            \State \emph{// Each hyperedge gathers its vertex neighbor labels.}
            \For{$v/v' \in \mathcal{V}/\mathcal{V}'$}
            \State \emph{// 1. Gathering vertex's hyperedge neighbor labels into multiset.}
            \State $M^e_i(v) \gets \{l^e_i(e) \mid e \in \mathcal{N}_e(v) \}$ 
            \State $M^e_i(v') \gets \{l^e_i(e') \mid e' \in \mathcal{N}_e(v') \}$ 
            \State \emph{// 2. Sorting multiset.}
            \State $S_v / S_{v'} \gets str(sort(M^e_i(v))) / str(sort(M^e_i(v')))$
            \State \emph{// 3. Vertex label compression and Relabeling.}
            \State $S_v / S_{v'} \gets strcat(l_{i-1}^v(v), S_v) / strcat(l_{i-1}^v(v'), S_{v'})$
            \State $l_{i}^v(v) / l_{i}^v(v') \gets f(S_v) / f(S_{v'})$
            \EndFor
            \EndWhile
            % \If{$i=h$ and $\{l_h^v(v) \mid v \in \mathcal{V}\} = \{l_h^v(v') \mid v' \in \mathcal{V}'\}$}
            % \State Output $True$
            % \Else
            % \State Output $False$
            % \EndIf
            \State Comparing $\{l_h^v(v) \mid v \in \mathcal{V}\}$ and $\{l_h^v(v') \mid v' \in \mathcal{V}'\}$
		\Ensure  Whether $\mathcal{G}$ and $\mathcal{G}'$ are isomorphism or not.
	\end{algorithmic}  
\end{algorithm}

The key idea of the series Weisfeiler-Lehman test algorithm is to augment the vertex labels by the sorted set of vertex labels of neighboring vertices and compress those augmented labels into new labels. These steps are repeated until the vertex label sets of the two input graphs differ or the number of iterations reaches $h$. The definition of neighbor is the biggest obstacle to extending it from graphs to hypergraphs. In simple graphs, each edge only links two vertices. Thus, we say that two vertices are neighbors to each other if there is an edge between them. However, the edge (hyperedge) can link more than two vertices in hypergraphs. The definition of neighbor relation is more difficult in hypergraphs. Motivated by \cite{hgnnp}, we adopt the hyperedge's vertex neighbors $\mathcal{N}_v(e)$ and vertex's hyperedge neighbors $\mathcal{N}_e(v)$ to represent the complex neighbor relation in hypergraphs. This way, the complex neighbor relation in a hypergraph is transferred hierarchically, and the hyperedge bridges the relationship between the vertex set and another, as shown in Figure \ref{fig:neighbor}. Given a hypergraph $\mathcal{G} = \{ \mathcal{V}, \mathcal{E} \}$, for each vertex $v \in \mathcal{V}$ its hyperedge neighbors $\mathcal{N}_e(v)$ and for each hyperedge $e \in \mathcal{E}$ its vertex neighbors $\mathcal{N}_v(e)$ can be formulated as
\begin{equation}
\label{eq:hg_nbr}
    \left\{
        \begin{aligned}
            \mathcal{N}_e(v)& = \{ e \mid \mathbf{H}(v,e) = 1, e \in \mathcal{E} \}\\
            \mathcal{N}_v(e)& = \{ v \mid \mathbf{H}(v,e) = 1, v \in \mathcal{V} \}
        \end{aligned}
    \right. ,
\end{equation}
where $\mathbf{H} \in \{0, 1\}^{|\mathcal{V}| \times |\mathcal{E}|}$ is the hypergraph incidence matrix. With the definition of two neighbor functions $\mathcal{N}_e(\cdot)$ and $\mathcal{N}_v(\cdot)$, we can easily quantify the complex beyond-pairwise correlation in hypergraphs.
In the following, we will introduce the generalized hypergraph Weisfiler-Lehman algorithm based on the two neighbor functions, as presented in Algorithm \ref{alg:hg_wl}.

Supposing we have two hypergraphs $\mathcal{G}=\{\mathcal{V}, \mathcal{E} \}$ and $\mathcal{G'} = \{ \mathcal{V}', \mathcal{E}' \}$, the goal of the algorithm is to test whether the two input hypergraphs are isomorphic. Assuming each vertex $v/v' \in \mathcal{V}/\mathcal{V'}$ in the two hypergraphs is associated with labels via vertex label map $\ell^v := v/v' \rightarrow c$. Similarly, the label of each hyperedge $e/e' \in \mathcal{E}/\mathcal{E}'$ is given by the hyperedge label map $\ell^e := e/e' \rightarrow c$. At first, the number of vertices must be the same ($|\mathcal{V}| = |\mathcal{V'}|$), or they are not isomorphic. In each iteration, the labels of vertices are organized using a multiset \cite{multiset}, which allows for multiple instances for each of its elements compared to a set from a mathematical perspective. The multisets $\{l_i^v(v) \mid v \in \mathcal{V}\}$ and $\{l_i^v(v') \mid v' \in \mathcal{V}'\}$ are the primary identifiers in $h$ times iterations. If they are all the same, the two hypergraphs are isomorphic; otherwise, they are not. The initial elements of the vertex/hyperedge label multiset are initialized by the original vertex/hypergraph map $\ell^v(\cdot)/\ell^e(\cdot)$. Each iteration includes two sub-processes: vertex labels $\rightarrow$ hyperedge labels and hyperedge labels $\rightarrow$ vertex labels. 

In the first sub-process (vertex labels $\rightarrow$ hyperedge labels), for each hyperedge $e \in \mathcal{E}$, we construct a multiset $M^v_i(e)$ with the labels from its vertex neighbor set $\mathcal{N}_v(e)$, and its elements are converted to strings to create a string $S_e$. The sorted order of the multisets ensures that all identical strings are mapped to the same number, as they occur in a consecutive block. To incorporate the root hyperedge information, we add a prefix $l^e_{i-1}(e)$ to the string $S_e$. Then, the sting $S_e$ is compressed with a compression function $f(\cdot)$ to generate a short label for each hyperedge $e$. Those short labels serve as the most recent labels for the hyperedges in the hypergraph. Note that the compression function $f(\cdot)$ must be an injective function, meaning that $f(l)=f(l')$ if and only if $l=l'$. In the second sub-process (hyperedge labels $\rightarrow$ vertex labels), the hyperedge labels' neighbor gathering, sorting, and compression are conducted similarly to the first sub-process. 

Finally, after $h$ iterations, the hypergraph Weisfiler-Lehman test algorithm can determine whether the two hypergraphs are isomorphic. Similar to the original Weisfeiler-Lehman algorithm \cite{graph_wl, graph_wl_valid}, the hypergraph Weisfeiler-Lehman can be applied to almost all hypergraphs (Section \ref{sec:exp:hypergraph} of Experiments). However, there are still some cases where it fails. Figure \ref{fig:alg_failed} illustrates an example in which the proposed hypergraph Weisfeiler-Lehman algorithm fails. In $i$-th iteration, the algorithm encodes the specified subtree with a height of $i$ for each vertex in the hypergraph. Thus, comparing the vertex label multiset $l_i^v(v)$  is equivalent to comparing the rooted subtree with a height of $i$ for the given vertex $v$. Additionally, in each iteration, the former rooted subtree information can be reused during the two-stage relabeling process. This approach avoids redundant computations involved in rebuilding the subtree with a specified height. Therefore, the devised hypergraph Weisfeiler-Lehman test algorithm is effective and efficient for the hypergraph isomorphism test.

\begin{figure}
    \centering
    \includegraphics[width=0.42\textwidth]{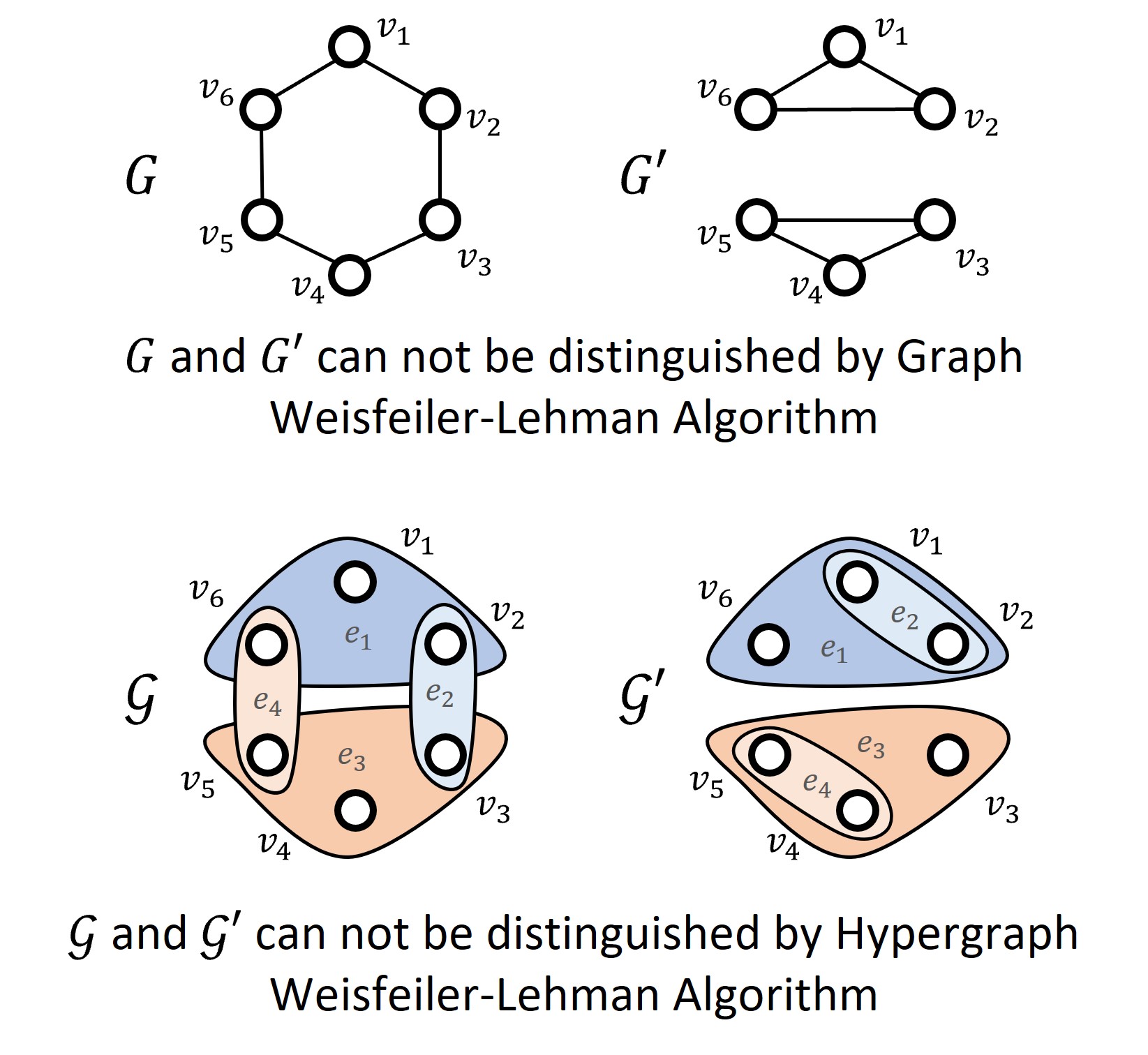}
    \caption{Examples of Graph/Hypergraph Weisfeiler-Lehman Algorithm Failed.}
    \label{fig:alg_failed}
\end{figure}

% \subsubsection{Complexity}
\textbf{Complexity. }
\label{sec:fw_complexity}
The runtime complexity of the proposed hypergraph Weisfeiler-Lehman test algorithm with $h$ iterations is $\mathcal{O}\left(hm\right)$. Here, $m$ is the capacity of the hypergraph, which can be computed as the number of non-zero elements in the hypergraph incidence matrix $\mathbf{H}$. The value of $m$ can be computed by either $m = |\mathcal{V}|\bar{d}_v$, or $m = |\mathcal{E}|\bar{d}_e$. $\bar{d}_v$ and $\bar{d}_e$ are the average degrees of vertices and hyperedges, which can be computed by $\bar{d}_v = \frac{1}{|\mathcal{V}|} \sum_{v \in \mathcal{V}} d_v$ and $\bar{d}_e = \frac{1}{|\mathcal{E}|} \sum_{e \in \mathcal{E}} d_e$, respectively.
% in the worst c $(|\mathcal{E}| \hat{d}_e + |\mathcal{V}| \hat{d}_v)$ $\hat{d}_e$ and $\hat{d}_v$ are the maximum degrees of the hyperedge and vertex, respectively. 
Here we first consider the subprocess of vertex labels $\rightarrow$ hyperedge labels. Clearly, the worst runtime of the neighbor gathering and stringify are $\mathcal{O}(\bar{d}_e)$. Due to the multiset only containing finite labels, we can use the radix sort algorithm\cite{radix_sort} to achieve the time complexity $\mathcal{O}(\bar{d}_e)$. The string concatenation, string compression, and the relabeling cost $\mathcal{O}(1)$ time complexity. Thus, the runtime of the first subprocess is $\mathcal{O}(|\mathcal{E}|\bar{d}_e)$ for all hyperedges. Similarly, the runtime of the second subprocess is $\mathcal{O}(|\mathcal{V}|\bar{d}_v)$ for all vertices. Hence, all those steps with $h$ iterations result in a total runtime of $\mathcal{O}\left(hm\right)$.

%用 “leftangle”....
% \begin{figure*}[!htbp]
%     \centering
%     \includegraphics[width=0.87\textwidth]{figs/alg.jpg}
%     \caption{Illustration of hypergraph Weisfeiler-Lehman kernel algorithm.}
%     \label{fig:alg}
% \end{figure*}
\begin{figure*}[!htbp]
    \centering
    \includegraphics[width=0.88\textwidth]{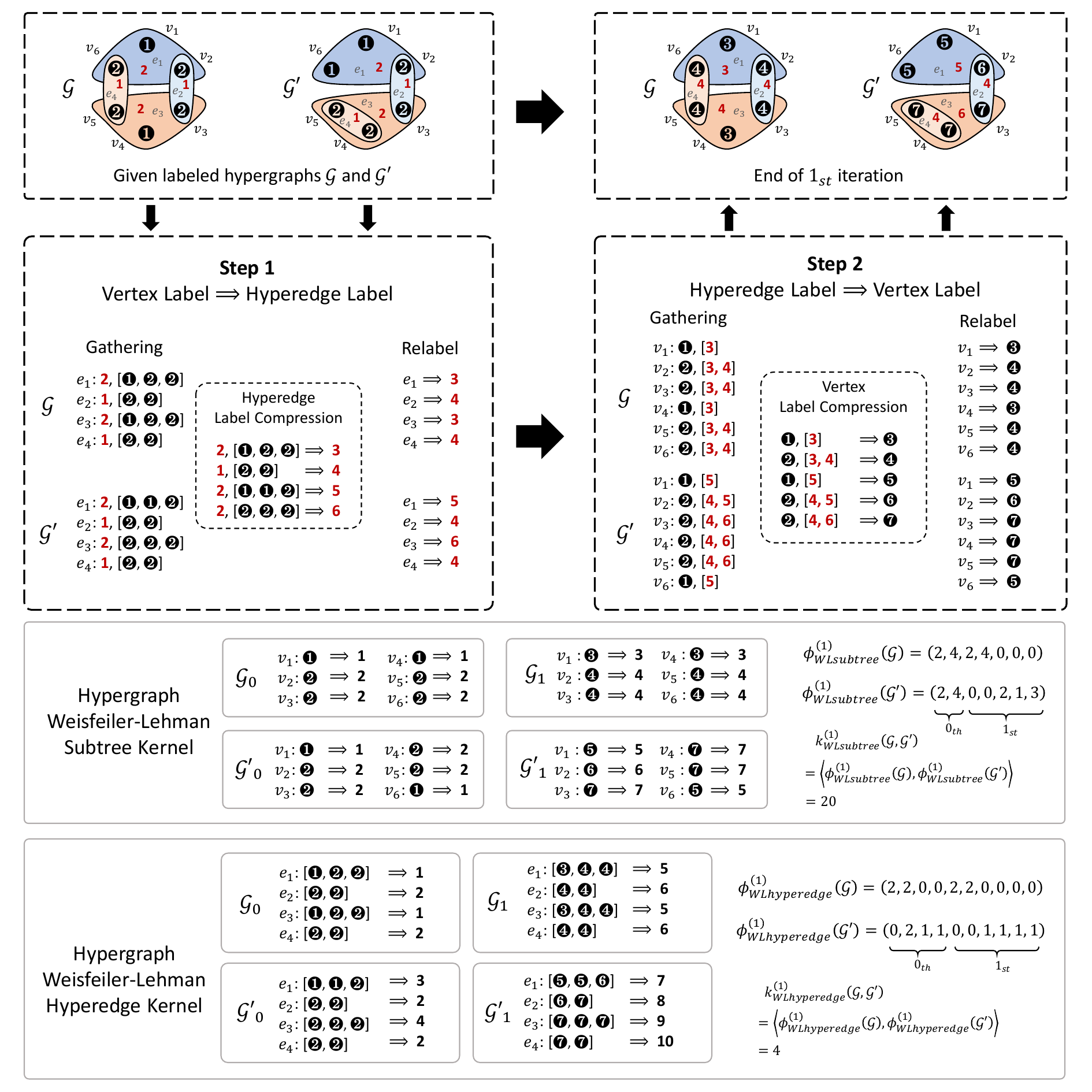}
    \caption{Illustration of hypergraph Weisfeiler-Lehman kernel algorithm.}
    \label{fig:alg}
\end{figure*}

\subsection{Hypergraph Weisfeiler-Lehman Kernel Framework}
In this subsection, we define the hypergraph Weifeiler-Lehman sequences and present the definition of the general hypergraph Weifeiler-Lehman kernel based on the proposed hypergraph Weifeiler-Lehman test algorithm. 

\subsubsection{Hypergraph Weisfeiler-Lehman Sequence}
In iteration $i$ of the hypergraph Weifeiler-Lehman algorithm, we generate two labeling functions $l_i^v$  and $l_i^e$ for vertices and hyperedges, respectively. Note that those labeling functions are concordant for the input hypergraph $\mathcal{G}$ and $\mathcal{G}'$. If two vertices in the hypergraph $\mathcal{G}$ or hypergraph $\mathcal{G}'$ share the same label, it indicates that these vertices possess identical rooted subtrees. Here, we treat each iteration (including two subprocesses) as a function $r((\mathcal{V}, \mathcal{E}, l_i^v, l_i^e)) = (\mathcal{V}, \mathcal{E}, l_{i+1}^v, l_{i+1}^e)$ that transforms the input hypergraphs in the same manner. After $h$ iterations, we obtain a sequence of $h+1$ hypergraphs (including the original hypergraph). This sequence is known as the hypergraph Weisfeiler-Lehman sequence, and it can be defined as follows.

\begin{definition} 
\label{def:hg_seq}
Given a hypergraph $\mathcal{G} = (\mathcal{V}, \mathcal{E}, \ell^v, \ell^e)$ and relabeling function $r := (\mathcal{V}, \mathcal{E}, l_i^v, l_i^e) \rightarrow (\mathcal{V}, \mathcal{E}, l_{i+1}^v, l_{i+1}^e)$, $\mathcal{G}_i$ indicates the relabeled hypergraph after $i$ iterations. Then, the \textbf{ hypergraph Weisfiler-Lehman sequence } can be defined as
\begin{equation}
    \{ \mathcal{G}_0, \cdots, \mathcal{G}_h \} = \{ (\mathcal{V}, \mathcal{E}, l_{0}^v, l_{0}^e), \cdots, (\mathcal{V}, \mathcal{E}, l_{h}^v, l_{h}^e) \}, 
    \nonumber
\end{equation}
where $\mathcal{G}_0 = \mathcal{G}$ and $l_0 = \ell_0$, the \textbf{ hypergraph Weisfeiler-Lehman sequence} up to height $h$ of hypergraph $\mathcal{G}$.
\end{definition}

In Definition \ref{def:hg_seq}, $\mathcal{G}_0$, $l_0^v$, and $l_0^e$ are initialized from the original hypergraph $\mathcal{G}$. The relabeling function $r := \mathcal{G}_{i-1} \rightarrow \mathcal{G}_i$ outcomes the relabeled hypergraph from the last iteration $\mathcal{G}_{i-1}$. In the sequence, the vertex set $\mathcal{V}$, and hyperedge set $\mathcal{E}$ are identical, while the vertex and hyperedge labels change in each iteration.

\subsubsection{General Hypergraph Weisfeiler-Lehman Kernel}
Given hypergraphs, we can obtain the hypergraph Weisfeiler-Lehman sequences by Algorithm \ref{alg:hg_wl}. To generalize the hypergraph embeddings from those sequences for computation, we then define the general hypergraph Weisfeiler-Lehman kernel. 

\begin{definition}
\label{def:hg_wl_fw}
Let $k$ be any kernel for hypergraphs that we will call the base kernel, given two hypergraphs $\mathcal{G}$ and $\mathcal{G}'$, the \textbf{hypergraph Weisfeiler-Lehman kernel} with $h$ iterations can be defined as 
\begin{equation}
    k_{\text{WL}}^{(h)}(\mathcal{G}, \mathcal{G}') = k(\mathcal{G}_0, \mathcal{G}_0') + k(\mathcal{G}_1, \mathcal{G}_1') + \cdots + k(\mathcal{G}_h, \mathcal{G}_h'),
\end{equation}
where $\{ \mathcal{G}_0, \mathcal{G}_1, \cdots, \mathcal{G}_h \}$ and $\{ \mathcal{G}_0', \mathcal{G}_1', \cdots, \mathcal{G}_h' \}$ are the \textbf{hypergraph Weisfeiler-Lehman sequence} of hypergraph $\mathcal{G}$ and $\mathcal{G}'$, respectively.
\end{definition}

Definition \ref{def:hg_wl_fw} presents a comprehensive framework for kernel-based hypergraph embedding, capable of incorporating discrete rooted subtrees with varying heights into consideration. When given an explicit strategy of the basic kernel, like label types counting, the hypergraph Weisfeiler-Lehman kernel will transform the original hypergraph structure into a vector embedding in the Hilbert space. This embedding representation enables various tasks, including hypergraph classification, retrieval, and similarity measurement, to be accomplished. 

\begin{theorem}
Let the base kernel $k$ be any positive semidefinite kernel on hypergraphs. Then, the corresponding hypergraph Weisfeiler-Lehman kernel $k_{\text{WL}}^{(h)}$ is positive semidefinite.
\end{theorem}

\begin{proof}
Let $\phi$ be the feature mapping corresponding to the kernel $k$:
$$
k(\mathcal{G}_i, \mathcal{G}_i') = \langle \phi(\mathcal{G}_i), \phi(\mathcal{G}_i') \rangle .
$$
We have 
$$
\mathcal{G}_i = \underbrace{r \cdots r}_{i}(\mathcal{G}) = R^i(\mathcal{G}) \quad \text{and} \quad \mathcal{G}_i' = \underbrace{r \cdots r}_{i}(\mathcal{G}') = R^i(\mathcal{G}').
$$
Thus, we have 
$$
k(\mathcal{G}_i, \mathcal{G}_i') = k(R^i(\mathcal{G}), R^i(\mathcal{G}')) = \langle \phi(R^i(\mathcal{G})), \phi(R^i(\mathcal{G}')) \rangle.
$$
We can build a composition function $\psi(\cdot)$ of $\phi(\cdot)$ and $R^i(\cdot)$. Then, we have
$$
k(\mathcal{G}_i, \mathcal{G}_i') = \langle \psi(\mathcal{G}), \psi(\mathcal{G}') \rangle .
$$
Since the sum of positive semidefinite kernel is a positive semidefinite kernel and $k$ is a positive semidefinite kernel of $\mathcal{G}$ and $\mathcal{G}'$, the $k_{\text{WL}}^{(h)}$ is positive semidefinite. 
\end{proof}

In the following, we will implement two basic kernel instances based on the general hypergraph Weisfeiler-Lehman kernel framework: subtree kernel and hyperedges kernel.

\subsection{Hypergraph Weisfeiler-Lehman Subtree Kernel}
Based on the Algorithm \ref{alg:hg_wl}, this subsection presents a natural instance named hypergraph Weisfeiler-Lehman subtree kernel. Since the vertex label in iteration $i$ compress a unique rooted subtree with height $i$, it is an intuitive strategy to use the frequency of discrete labels to describe the original hypergraph structure. By generalizing the Weisfeiler-Lehman subtree kernel \cite{graph_wl_subtree} from graphs to hypergraphs, we define the \textbf{hypergraph Weisfeiler-Lehman subtree kernel} as follows. 

\begin{definition}
\label{def:hg_wl_subtree}
    Let $\mathcal{G}$ and $\mathcal{G}'$ be hypergraphs. Define $\Sigma_i \subseteq \Sigma $ as the vertex label set that occurs in iteration $i$ of the hypergraph Weisfeiler-Lehman algorithm. Let $\{ \mathcal{G}_0, \mathcal{G}_1, \cdots, \mathcal{G}_h\}$ and $\{ \mathcal{G}_0', \mathcal{G}_1', \cdots, \mathcal{G}_h'\}$ be the hypergraph Weisfeiler-Lehman sequences of hypergraph $\mathcal{G}$ and $\mathcal{G}'$, respectively. Let $\Sigma_0$ be the original label set from the hypergraph $\mathcal{G}$ and $\mathcal{G}'$. Assume all $\Sigma_i$ are pairwise disjoint. Without loss of generality, assume that every $\Sigma_i = \{ \sigma_i^1, \sigma_i^2, \cdots, \sigma_i^{|\Sigma_i|} \}$ is ordered. Define a map $c:= \{ \mathcal{G}_i, \mathcal{G}_i' \} \times \sigma_i^j \rightarrow \mathbb{N}, i \in [0, 1, \cdots, h]$, such that $c(\mathcal{G}_i, \sigma_i^j)$ is the number of occurrences of the label $\sigma_i^j$ in the hypergraph $\mathcal{G}_i$. Then, the Weisfeiler-Lehman subtree kernel on two hypergraph $\mathcal{G}$ and $\mathcal{G}'$ with $h$ iterations is defined as
    \begin{equation}
        k_{\text{WLsubtree}}^{(h)}(\mathcal{G}, \mathcal{G}') = \langle \phi_{\text{WLsubtree}}^{(h)}(\mathcal{G}), \phi_{\text{WLsubtree}}^{(h)}(\mathcal{G}') \rangle ,
    \end{equation}
    where
    \begin{small}
    \begin{equation}
        \phi_{\text{WLsubtree}}^{(h)}(\mathcal{G}) = \left( c(\mathcal{G}_0, \sigma_0^{1}), \cdots, c(\mathcal{G}_0, \sigma_0^{|\Sigma_0|}), \cdots, c(\mathcal{G}_h, \sigma_h^{|\Sigma_h|}) \right) ,
        \nonumber
    \end{equation}
    \end{small}
    and 
    \begin{small}
    $$
    \phi_{\text{WLsubtree}}^{(h)}(\mathcal{G}') = \left( c(\mathcal{G}_0', \sigma_0^{1}), \cdots, c(\mathcal{G}_0', \sigma_0^{|\Sigma_0|}), \cdots, c(\mathcal{G}_h', \sigma_h^{|\Sigma_h|}) \right) .
    $$
    \end{small}
\end{definition}

% \begin{theorem}
%     Given connected hypergraphs $\mathcal{G}$ and $\mathcal{G}'$, let $\Sigma = \{ \Sigma_0, \Sigma_1, \cdots, \Sigma_h \}$ be the label sets in iterations of the hypergraph Weisfeiler-Lehman subtree kernel. $\Sigma_i = \{ \sigma_i^1, \sigma_i^2, \cdots, \sigma_i^{|\Sigma_i|} \}$ is the ordered label set of iteration $i$. Assume with the compression function $f:= s \rightarrow s'$, the length of the compressed string $s'$ is strictly shorter than the length of the input string $s$. Then, $\Sigma_i$ and $\Sigma_j$ are disjoint for any $i, j \in [0, 1, \cdots, h]$.
% \end{theorem}

% \begin{proof}
%     Suppose there is an intersection between $\Sigma_i$ and $\Sigma_j$, and the element $\sigma_* = \sigma_i^a = \sigma_j^b$ belongs to the intersection, where $\sigma_i^a \in \Sigma_i$ and $\sigma_j^b \in \Sigma_j$. Without loss of generality, assume that $i < j$. 

%     If $j = i + 1$, according to the hypergraph Weisfeiler-Lehman algorithm, the string compression function is an injective function, the
%     label $\sigma^k_j, k\in [1, 2, \cdots, |\Sigma_j|]$ from $\Sigma_j$ can be calculated by $l_j^v(v) = f(strcat(l_i^v(v), S_v)), v \in \mathcal{V}$. balabala

%     If $ i > 0 $ and $j = i + 1$,
% \end{proof}

For the convenience of computation, the hypergraph Weisfeiler-Lehman subtree kernel can be further written as

\begin{small}
\begin{equation}
    \begin{aligned}
        k_{\text{WLsubtree}}^{(h)} (\mathcal{G}, \mathcal{G}') &= k(\mathcal{G}_0, \mathcal{G}_0') + k(\mathcal{G}_1, \mathcal{G}_1') + \cdots + k(\mathcal{G}_h, \mathcal{G}_h') \\
        &= \sum_{i=0}^{h} \sum_{v \in \mathcal{V}} \sum_{v' \in \mathcal{V}'} \delta \left( l_i^v(v), l_i^v(v') \right) \\
        &= \sum_{i=0}^{h} \sum_{j=1}^{|\Sigma_i|} c(\mathcal{G}_i, \sigma_i^j) \cdot c(\mathcal{G}_i', \sigma_i^j)
    \end{aligned} ,
\end{equation}
\end{small}
where $\delta(a, b)$ is the Dirac kernel, that is, it is $1$ when $a$ and $b$ are equal and $0$ otherwise. $c(\mathcal{G}_i, \sigma_i^j)$ is the element counting function as defined in Definition \ref{def:hg_wl_subtree}.

% \subsubsection{Complexity}
\textbf{Complexity. }
\label{sec:subtree_complexity}
The runtime complexity of the proposed hypergraph Weisfeiler-Lehman subtree kernel on $N$ hypergraphs with $h$ iterations is $\mathcal{O}\left(Nh\bar{m} + N^2hd\right)$, where $\bar{m}$ is the average hypergraph capacity (referring to Section \ref{sec:fw_complexity}) of $N$ hypergraph. $d$ denotes the dimension of the final hypergraph feature, which equals the sum of $\{|\Sigma_0|, |\Sigma_1|, \cdots, |\Sigma_h| \}$. The computation of the hypergraph Weisfeiler-Lehman subtree kernel can be divided into two stages: building the hypergraph Weifeiler-Lehman sequence and counting the vertex labels for each iteration. Clearly, referring to Section \ref{sec:fw_complexity}, the complexity of the first stage is $\mathcal{O}(Nh\bar{m})$. As for the second stage, we multiply all feature vectors to get the kernel values of all hypergraph pairs. Since the dimension of each feature vector is $d$, the computation of the second stage requires a runtime $\mathcal{O}(N^2hd)$. Hence, all those steps result in a total runtime of $\mathcal{O}\left(Nh\bar{m} + N^2hd\right)$.

\subsection{Hypergraph Weisfeiler-Lehman Hyperedge Kernel}
This subsection presents another instance of the general hypergraph Weisfeiler-Lehman hyperedge kernel: hypergraph Weisfeiler-Lehman hyperedge kernel. Unlike the subtree kernel, this instance counts the frequency of different types of hyperedges. Since each hyperedge can link more than two vertices, we treat it as an ordered set for comparison. Compared with the previous subtree kernel, the hyperedge kernel can directly represent the information of high-order connections in each iteration.
It can be defined as follows.
\begin{definition}
    Let $\mathcal{G}$ and $\mathcal{G}'$ be hypergraphs. Define $\Sigma_i \subseteq \Sigma $ as the vertex label set that occurs in iteration $i$ of the hypergraph Weisfeiler-Lehman algorithm. Let $\{ \mathcal{G}_0, \mathcal{G}_1, \cdots, \mathcal{G}_h\}$ and $\{ \mathcal{G}_0', \mathcal{G}_1', \cdots, \mathcal{G}_h'\}$ be the hypergraph Weisfeiler-Lehman sequences of hypergraph $\mathcal{G}$ and $\mathcal{G}'$, respectively. Let $\Sigma_0$ be the original label set from the hypergraph $\mathcal{G}$ and $\mathcal{G}'$. Assume all $\Sigma_i$ are pairwise disjoint. 
    Define map $z_i := e \rightarrow (l_i^v(v_1), \cdots, l_i^v(v_{|e|})), v_j \in e, l_i^v(v_j) \in \Sigma_i, e \in \mathcal{E}$ as the hyperedge code to denote the hyperedge in $i$-th iteration. Assume each hyperedge code is a sorted vertex label tuple. Define $ \Omega_i \subseteq \Omega $ as the hyperedge code set that occurs in iteration $i$. 
    Without loss of generality, assume that every $\Sigma_i = \{ \sigma_i^1, \sigma_i^2, \cdots, \sigma_i^{|\Sigma_i|} \}$ and $\Omega_i = \{ \omega_i^1, \omega_i^2, \cdots \omega_i^{|\Omega_i|} \}$ are ordered. 
     Define a map $p:= \{ \mathcal{G}_i, \mathcal{G}_i' \} \times \omega_i^j \rightarrow \mathbb{N}, i \in [0, 1, \cdots, h]$, such that $p(\mathcal{G}_i, \omega_i^j)$ is the number of occurrences of the hyperedge code $\omega_i^j$ in the hypergraph $\mathcal{G}_i$. Then, the Weisfeiler-Lehman hyperedge kernel on two hypergraph $\mathcal{G}$ and $\mathcal{G}'$ with $h$ iterations is defined as
    \begin{equation}
        k_{\text{WLhyperedge}}^{(h)}(\mathcal{G}, \mathcal{G}') = \langle \phi_{\text{WLhyperedge}}^{(h)}(\mathcal{G}), \phi_{\text{WLhyperedge}}^{(h)}(\mathcal{G}') \rangle ,
    \end{equation}
    where
    \begin{small}
    \begin{equation}
        \phi_{\text{WLhyperedge}}^{(h)}(\mathcal{G}) = \left( p(\mathcal{G}_0, \omega_0^{1}), \cdots, p(\mathcal{G}_0, \omega_0^{|\Omega_0|}), \cdots, p(\mathcal{G}_h, \omega_h^{|\Omega_h|}) \right) ,
        \nonumber
    \end{equation}
    \end{small}
    and 
    \begin{small}
    $$
    \phi_{\text{WLhyperedge}}^{(h)}(\mathcal{G}') = \left( p(\mathcal{G}_0', \omega_0^{1}), \cdots, p(\mathcal{G}_0', \omega_0^{|\Omega_0|}), \cdots, p(\mathcal{G}_h', \omega_h^{|\Omega_h|}) \right) .
    $$
    \end{small}
\end{definition}

Similarly, the hypergraph Weisfeiler-Lehman hyperedge kernel can be further written as 
\begin{small}
\begin{equation}
\label{eq:k_hyperedge}
    \begin{aligned}
        k_{\text{WLhyperedge}}^{(h)} (\mathcal{G}, \mathcal{G}') &= k(\mathcal{G}_0, \mathcal{G}_0') + \cdots + k(\mathcal{G}_h, \mathcal{G}_h') \\
        &= \sum_{i=0}^{h} \sum_{e \in \mathcal{E}} \sum_{e' \in \mathcal{E}'} \delta \left( z_i(e), z_i(e') \right) \\
        &= \sum_{i=0}^{h} \sum_{j=1}^{|\Omega_i|} p(\mathcal{G}_i, \omega_i^j) \cdot p(\mathcal{G}_i', \omega_i^j)
    \end{aligned} ,
\end{equation}
\end{small}
where $\delta(a, b)$ is the Dirac kernel.

% \subsubsection{Complexity}
\textbf{Complexity. }
The runtime complexity of the proposed hypergraph Weisfeiler-Lehman hyperedge kernel on $N$ hypergraphs with $h$ iterations is $\mathcal{O}\left(Nh\bar{m} + N^2hd'\right)$. $d'$ is the dimension of the final feature, which equals the sum of $\{|\Omega_0|, |\Omega_1|, \cdots, |\Omega_h| \}$. As Section \ref{sec:subtree_complexity}, the runtime of building the hypergraph Weisfeiler-Lehman sequence is $\mathcal{O}(Nh\bar{m})$. For each hypergraph, the runtime of building hyperedge codes with $h$ iteration is $\mathcal{O}(h|\mathcal{E}|\bar{d}_v)$, which equals $\mathcal{O}(hm)$. Thus, for $N$ hypergraph, the runtime of building hyperedge codes is $\mathcal{O}\left(Nh\bar{m} \right)$. As for the hyperedge code counting stage, the complexity of pairs of $N$ hypergraphs is $\mathcal{O}(N^2hd')$ as shown in Equation \eqref{eq:k_hyperedge}. Hence, all those steps result in a total runtime of $\mathcal{O}\left(Nh\bar{m} + N^2hd'\right)$.

\subsection{Computing in Practice}
Due to the large and uncertain dimensions of the generated feature from the hypergraph Weisfeiler-Lehman kernel, we adopt the commonly used normalization trick \cite{svm_trick} to reduce the dimension of the feature for training and testing hypergraphs. Given the feature matrices $\mathbf{X}_{tr} \in \mathbb{R}^{N_{tr} \times d}, \mathbf{X}_{te} \in \mathbb{R}^{N_{te} \times d}$ extracted with a specifical hypergraph Weisfeiler-Lehman kernel $K$, the $N_{tr}$ and $N_{te}$ denote the number of training hypergraphs and the testing hypergraphs, respectively. We compress them based on the training hypergraphs with the function $Z:= \mathbb{R}^{N \times d} \rightarrow \mathbb{R}^{N \times N_{tr}}$. The compressed features can be computed by $\check{\mathbf{X}}_{tr} = Z(\mathbf{X}_{tr}) \in \mathbb{R}^{N_{tr} \times N_{tr}}$ and $\check{\mathbf{X}}_{te} = Z(\mathbf{X}_{te}) \in \mathbb{R}^{N_{te} \times N_{tr}}$, respectively. Each entry of the compressed features $\check{\mathbf{X}}_{tr}$ and $\check{\mathbf{X}}_{te}$ can be computed as
\begin{equation}
    \check{\mathbf{X}}_{i,j} = \frac{k(\mathcal{G}_i, \mathcal{G}_j)}{\sqrt{k(\mathcal{G}_i, \mathcal{G}_i) k(\mathcal{G}_j, \mathcal{G}_j)}} , \qquad j = 1, \cdots, N_{tr} ,
\end{equation}
where the $i$ is in range $[1, N_{tr}]$ or $[1, N_{te}]$ for training and testing, respectively. The function $Z$ can be defined as
\begin{equation}
    Z(\mathbf{X}) = \frac{\mathbf{X} \mathbf{X}_{tr}^\top}{\sqrt{(\mathbf{X} \circ \mathbf{X} \mathbf{1}_d) \otimes (\mathbf{X}_{tr} \circ \mathbf{X}_{tr} \mathbf{1}_d)^\top}} ,
\end{equation}
where $\circ, \otimes$ denote the pointwise product, and outer product. $\mathbf{1}_d \in \{1\}^{d \times 1}$ is the vector of ones. Note that in the embedding process of the testing hypergraphs, those unseen structures like rooted subtrees and hyperedges are dropped. In this way, the final features from training hypergraphs and testing hypergraphs are comparable to each other.

% \subsection{Analysis}

\subsection{Relation to Graph WL Subtree Kernel}
\label{sec:relation_graph_wl_and_hypergraph_wl}
In this subsection, we prove that the hypergraph Weisfeiler-Lehman subtree kernel can degenerate into the graph Weisfeiler-Lehman subtree kernel confronting the simple graph structure.

\begin{lemma}
    \label{lemma:bijective}
    Given a graph $G = \{ V, E \}$, let the $c$ and $c'$ be the graph/hypergraph Weisfeiler-Lehman subtree kernel's compression functions compress the rooted string $s$/$s'$ to a unique label $l$/$l'$, respectively. Exist a bijective function $\phi$ that maps the graph Weisfeiler-Lehman subtree kernel's rooted string $s$ to hypergraph Weisfeiler-Lehman kernel' rooted string $s'$. 
\end{lemma}
% 双射 -> 集合元素数量相等

\begin{proof}
    \textbf{Proof of surjection:} In the $i$-th iteration, given vertex $v$, the compressed rooted string of the graph Weisfeiler-Lehman subtree kernel \cite{graph_wl_subtree} is $l_{i-1}(v)|\langle l_{i-1}(v_j), v_j \in \mathcal{N}(v)\rangle$, where $\mathcal{N}(v) = \{v_1, v_2, \cdots\}$ is the neighbor vertex set of the vertex $v$. $\langle \cdot \rangle$ is the multi-set sorting function. $\cdot|\cdot$ is the string concatenation operation. Finally, in $i$-th iteration, the rooted string of the graph Weisfeiler-Lehman subtree kernel of the vertex $v$ can be formulated as
    \begin{equation}
    \label{eq:graph_rooted_string}
        l_{i-1}(v)|\langle l_{i-1}(v_1), l_{i-1}(v_2), \cdots, l_{i-1}(v_m) \rangle .
    \end{equation}
    By composing label function $l(\cdot)$, we can construct an injective function $f$ that maps an order vertex set to the rooted string. Then, the rooted string of vertex $v$ can be written as $f(v, v_1, v_2, ..., v_m)$.

    Similarly, given the vertex $v$, in the $i$-th iteration, the compressed rooted string of the hypergraph Weisfeiler-Lehman subtree kernel is $l_{i-1}^v(v)|\langle l_i^e(e_j), e_j \in \mathcal{N}_e(v)\rangle$, where $\mathcal{N}_e(v) = \{ e_1, e_2, \cdots \}$ is the hyperedge neighbor set of the given vertex $v$. Then, each hyperedge label can be computed by $l^e_i(e) = l^e_{i-1}(e) | \langle l^v_{i-1}(v_j), v_j \in \mathcal{N}_v(e) \rangle$, where $\mathcal{N}_v(e) = \{ v_1, v_2, \cdots \}$ is vertex neighbor set of the given hyperedge $e$. Finally, in $i$-th iteration, the rooted string of the hypergraph Weisfeiler-Lehman subtree kernel of the vertex $v$ can be formulated as
    \begin{equation}
        \label{eq:initial_hg_rooted_string}
        \begin{aligned}
        l_{i-1}^v(v)|\langle & l_{i-1}^e(e_1)|\langle l^v_{i-1}(v^{e_1}_1), l^v_{i-1}(v^{e_1}_2), \cdots \rangle,  \\
        & l_{i-1}^e(e_2)|\langle l^v_{i-1}(v^{e_2}_1), l^v_{i-1}(v^{e_2}_2), \cdots \rangle,  \\
        &\cdots \\
        & l_{i-1}^e(e_m)|\langle l^v_{i-1}(v^{e_m}_1), l^v_{i-1}(v^{e_m}_2), \cdots \rangle \rangle
        \end{aligned} ,
    \end{equation}
    where $\{ v^{e_j}_1, v^{e_j}_2 \}$ denotes the vertex set that linked by hyperedge $e_j$. Considering that each edge in the given graph $G$ only connects two vertices and one of the connected vertex must be the vertex $v$, the Equation \eqref{eq:initial_hg_rooted_string} can be further written as:
    \begin{equation}
        \label{eq:hg_rooted_string}
        \begin{aligned}
        l_{i-1}^v(v)|\langle & l_{i-1}^e(e_1)|\langle l^v_{i-1}(v), l^v_{i-1}(v_1) \rangle,  \\
        & l_{i-1}^e(e_2)|\langle l^v_{i-1}(v), l^v_{i-1}(v_2) \rangle,  \\
        &\cdots \\
        & l_{i-1}^e(e_m)|\langle l^v_{i-1}(v), l^v_{i-1}(v_m) \rangle \rangle
        \end{aligned} .
    \end{equation}
    Since the degree of the hyperedge in the simple graph is constant to $2$, the hyperedge label $l^e_0(e)$ is all the same: $l^e_0(e_1) = l^e_0(e_2) = \cdots = l^e_0(e_m) $.

    When $i=0$, the two algorithms both count the number of different types of labels, thus yielding the same number of rooted strings. 

    When $i=1$, by composing the label function $l^v(\cdot)$ and initial hyperedge label $l^e_0(e)$, we can construct an inject function $f'$ that map the vertex set to string. Then, the rooted string of the vertex $v$ extracted from the hypergraph Weisfeiler-Lehman can be written as $f'(v, v_1, v_2, \cdots, v_m)$, which can be mapped the rooted string extracted from the graph Weisfeiler-Lehman. Since the input graph is the same and the number of the rooted subtree of height $1$ is also the same, the number of the unique labels in iteration $1$ of the two algorithms is the same.

    When $i>1$, the $l_{i-1}^e(e)$ is a function with respect to the $\{ l_{i-2}^e(e), l_{i-2}^v(v), l_{i-2}^v(v') \}$, where $v$ and $v'$ is the specified root vertex and another vertex that linked by edge $e$, respectively. For the specified vertex $v$, the related vertex set $\{v, \{v, v_1\}, \{v, v_2\}, \cdots, \{v, v_m\}\}$ that determining the rooted string can then be reduced to $\{v, v_1, v_2, \cdots, v_m\}$. We can also construct an injective function $f'$ for mapping. Then, the rooted string can be written as $f'(v, v_1, v_2, \cdots, v_m)$. Hence, in $i>1$ iteration, the number of the compressed label of the two algorithms are also the same. 

    Since the two algorithms yield the same number of unique compressed labels in each iteration, the number of the rooted strings extracted from the two algorithms with $h$ iterations is the same. Those rooted strings can also be determined from the same sorted vertex set. Thus, the surjection holds.
    
    \textbf{Proof of injection:} First, we construct a map $\phi := s \rightarrow s'$ from the graph Weisfeiler-Lehman rooted string to the hypergraph Weisfeiler-Lehman rooted string. Given a rooted string $s$ like Equation \ref{eq:graph_rooted_string}, we first extract the stem vertex set $\{ v, v_1, v_2, \cdots, v_m \}$ with injective function $f$. Note that the first element of the stem vertex set is the root vertex $v$ and the rest is the neighbor vertex set of the vertex $v$. Then, we transform the stem vertex set into an expanded set $\{v, \{v, v_1\}, \{v, v_2\}, \cdots, \{v, v_m\}\}$. In the rooted string of the hypergraph Weisfeiler-Lehman subtree kernel like Equation \eqref{eq:hg_rooted_string}, the updated hyperedge label $l_{i-1}^e(e_j)|\langle l^v_{i-1}(v), l^v_{i-1}(v_j)$ is a function of the $\{v, v_j\}$ and the initial hyperedge label $\ell(e)$ is all the same. Hence, the rooted string $s'$ can also be computed from the stem vertex set $\{ v, v_1, v_2, \cdots, v_m \}$, which can be denoted by the injective function $g$. Then, the map $\phi :\ s \rightarrow s'$ can be implemented as a composition function $ \phi = g \circ f$. 
    
    Next, we prove $\forall s_1, s_2 \in S, \phi(s_1) = \phi(s_2) \Rightarrow s_1 = s_2$.  $S$ denotes the final rooted string set extracted by the graph Weisfeiler-Lehman subtree kernel. From the definition of the algorithm \cite{graph_wl_subtree}, we know that any two elements in the $S$ are different and indicate different rooted subtree structures. Assume we have two rooted strings $s_1$ and $s_2$ that $\phi(s_1) = \phi(s_2)$, then the root vertex of the $\phi(s_1)$ and $\phi(s_2)$ are the same, which is termed by vertex $v$. Clearly, the operation of the injective function $g$ is reversible. Then the two same rooted strings $\phi(s_1)$ and $\phi(s_2)$ can be transformed into the stem vertex set $\{v, \{v, v_1\}, \{v, v_2\}, \cdots, \{v, v_m\}\}$. After removing the root vertex for each element and applying the inverse function of $f$, the $\phi(s_1)$ and $\phi(s_2)$ will produce the same rooted string $s_1$ and $s_2$. Since the assumption to be false leads to a contradiction, it is concluded that the injection holds.

    Finally, since the surjection and injection hold, according to the definition of the bijection, Lemma \ref{lemma:bijective} holds.
\end{proof}

\begin{theorem}
\label{theorem:equal}
    Given a set of graph $\mathcal{G} = \{ G_1, G_2, \cdots, G_n \}$, the kernel matrix $\mathbf{K} \in \mathbb{R}^{n \times n}$ and $\mathbf{K}' \in \mathbb{R}^{n \times n}$ are generated by the graph Weisfeiler-Lehman subtree kernel and the hypergraph Weisfeiler-Lehman subtree kernel, respectively. Then, the $\mathbf{K}$ and $\mathbf{K}'$ are identical. 
\end{theorem}

\begin{proof}
    Given graph $G_i \in \mathcal{G}$, the feature vector $x_i \in \mathbb{R}^{c \times 1}$ and $x_i' \in \mathbb{R}^{c \times 1}$ are the graph embeddings extracted from the graph/hypergraph Weisfeiler-Lehman subtree kernel, respectively. Since each element in the feature vector corresponds to a unique rooted string and with Lemma \ref{lemma:bijective}, there exists a permutation matrix $\mathbf{P} \in \{0, 1\}^{N \times N}$, such that $x_i' = \mathbf{P} x_i$. The entries of kernel matrix $\mathbf{K}$ and $\mathbf{K}'$ can be calculated by:
    \begin{equation}
        \left\{
            \begin{aligned}
                \mathbf{K}_{ij} &= k(G_i, G_j) = \langle x_i, x_j \rangle = x_i^\top  x_j \\
                \mathbf{K}'_{ij} &= k'(G_i, G_j) = \langle x'_i, x'_j \rangle = {x'}_i^\top x'_j \\
            \end{aligned}
        \right. ,
    \end{equation}
    where $\langle \cdot, \cdot \rangle$ denotes the inner product operation of two vectors. According to the permutation matrix, we have 
    \begin{equation}
        \begin{aligned}
        {x'}_i^\top x'_j &= (\mathbf{P} x_i)^\top (\mathbf{P} x_j) \\
        &= x_i^\top \mathbf{P}^\top \mathbf{P} x_j \\
        &= x_i^\top x_j
        \end{aligned} .
    \end{equation}
    Hence, each element in position $i, j$ of kernel matrix $\mathbf{K}$ and $\mathbf{K}'$ are identical. The theorem holds.
\end{proof}

% \subsubsection{Relation to HGNN+ Framework}
% balabala

%% file: 5_exp.tex
\section{Experiments}
In this section, we conduct experiments on seven graph and $12$ hypergraph datasets, which consist of both synthetic and real-world datasets. Additionally, We include an ablation study to compare the two proposed types of hypergraph Weisfeiler-Lehman kernels. Furthermore, we compare the computation runtime and robustness of these methods when facing data noising challenges.

\subsection{Common Settings}
This subsection outlines the standard experimental settings applied across all experiments, which may be overridden in each specified experiment. In the following subsections, the abbreviation \textit{WL} is utilized to represent \textit{Weisfeiler-Lehman}.

% \subsubsection{Compared Methods}
\textbf{Compared Methods. }
Since we focus on pure kernel methods on correlation data, for a fair comparison, we select four typical graph/hypergraph kernel-based methods.
%, including GraphLet Kernel \cite{graphlet}, Graph WL Subtree Kernel \cite{graph_wl_subtree}, Hypergraph Rooted Kernel \cite{hg_root}, and Hypergraph Directed Line Kernel \cite{hg_line}. Hypergraph WL Subtree Kernel and Hypergraph WL Hyperedge Kernel are the proposed methods.  
\textbf{GraphLet Kernel \cite{graphlet}.} A typical sampling-based graph kernel method. It draws and counts small factor graphs from a given graph. 
\textbf{Graph WL Subtree Kernel \cite{graph_wl_subtree}.} The method is naturally generalized from the Graph Weisfiler-Lehman algorithm, which counts the rooted subtrees of a given graph. 
\textbf{Hypergraph Rooted Kernel \cite{hg_root}.} A typical sampling-based hypergraph kernel method. It draws and counts hyper-paths from a given hypergraph.
\textbf{Hypergraph Directed Line Kernel \cite{hg_line}.} Hypergraph Directed Line Kernel is a transformation-based method that transforms the hypergraph into the directed graph. Then, the directed Weisfiler-Lehman algorithm is applied to extract the feature vector of a given hypergraph.

\subsubsection{Training Details}
For a fair comparison, we set the random seed in the $[2021, 2025]$ range, and the 5-fold cross-validation is adopted for each experiment. We report the average results of five folds and five seeds for all methods. Since those kernel methods can only generate the graph/hypergraph features, we utilize the standard SVM \cite{SVM, SVM_lib} as the classifier to compare those methods. Note that the SVM is initialized with the same hyper-parameters for all methods.% More details are provided in Appendix xxx.

\subsubsection{Other Details}
To assess the efficacy of capturing distinct correlation structures, we exclude the original vertex features of all datasets. Each vertex label is initialized based on the vertex degree, and the edge/hyperedge label is initialized according to the edge/hyperedge degree. Note that the edge label in graphs remains constant at two. 

Since traditional graph kernel methods cannot be directly applied to hypergraphs, we employ the common technique of clique expansion \cite{hgnnp} to transform the hypergraphs into graphs. As for the hypergraph kernels confronting the graph datasets, we consider these graphs as 2-uniform hypergraphs directly.

For the single-label classification task, we evaluate the accuracy (Acc) and macro F1 score (F1\_ma) as performance metrics. For the multi-label classification task, we compare the exact match ratio (EMR) and example-based accuracy (EB-Acc) as evaluation measures.

\subsection{Experiments on Graph Datasets}
This subsection presents the experiments on graph datasets, including two synthetic and five real-world graph datasets. The statistics of these datasets are summarized in Table \ref{tab:stat_graph}.

\begin{table}%[!htbp]
\centering
\caption{Statistic information of seven graph datasets.}
\label{tab:stat_graph}
\begin{threeparttable}
\begin{tabular}{ccccc}
\toprule
\multirow{2}{*}{Dataset} & \multirow{2}{*}{\# graphs} & \multirow{2}{*}{\# classes} & Avg. & Avg. \\
 &  &  & \# vertices & \# edges \\
 \midrule
RG-Macro & 1000 & 6 & 27.48 & 47.7 \\
RG-Sub & 1000 & 6 & 27.48 & 47.7 \\
IMDB-BINARY & 1000 & 2 & 19.8 & 96.5 \\
IMDB-MULTI & 1500 & 3 & 13.0 & 65.9 \\
MUTAG & 188 & 2 & 17.9 & 19.8 \\
NCI1 & 4110 & 2 & 29.9 & 32.3 \\
PROTEINS & 1113 & 2 & 39.1 & 72.8 \\
\bottomrule
\end{tabular}
    \begin{tablenotes}
        \item ``\#'' denotes ``number of'', and ``Avg.'' is short for ``Average''.
    \end{tablenotes}
\end{threeparttable}
\end{table}

\subsubsection{Datasets}

\textit{Synthetic Graph Datasets:} We have devised and generated two synthetic graph datasets, RG-Macro and RG-Sub, to test different algorithms' abilities on diverse low-order structures. We employed a two-step approach to create these graphs. First, we handpicked six widely-used graph structures as "subgraph" factors, including complete, bipartite, circle, cube, star, and wheel graphs. Randomly picked constructor parameters generate each graph. As an illustration, a circle graph characterized by a parameter $ \delta = 5 $ comprises a set of five interconnected nodes, with each node being connected to its neighboring node by a single edge. Second, we devise six ways to combine the "subgraph" factor graphs to produce a larger graph, such as chain linking, star linking, circle linking, and more. These linking strategies were recorded as a "macro" structure of the graph. Both datasets contain the same graphs, but the target label is different. The RG-Macro dataset uses the "macro" graph structure as a label, while the RG-Sub dataset tests the ability to differentiate between "subgraph" structures. %For further details and analysis, please see the Appendix xxxx.

\textit{Real-world Graph Datasets:} IMDB-BINARY\cite{IMDB} and IMDB-MULTI\cite{IMDB} are movie collaboration datasets. Each vertex in the two datasets is an actor. Each movie corresponds to a graph and associates a genre as its classification target. If two actors exist simultaneously in another movie, an edge will be created to link them. IMDB-BINARY contains movies from two genres: Action and Romance. IMDB-MULTI is a multi-class version of IMDB-BINARY and contains Comedy, Romance, and Sci-Fi genres. NCI1\cite{NCI1}, MUTAG\cite{MUTAG}, and PROTEINS\cite{PROTEINS} are bioinformatics datasets. NCI1 is collected from the National Cancer Institute (NCI) dataset. It is a balanced subset of chemical compounds screened for suppressing or inhibiting the growth of a panel of human tumor cell lines, with 37 discrete labels. MUTAG is a mutagenic aromatic and heteroaromatic nitro compounds dataset with seven discrete labels. PROTEINS is a dataset where nodes are secondary structure elements (SSEs), and there is an edge between two nodes if they are neighbors in the amino-acid sequence or 3D space. It has three discrete labels, representing helix, sheet, or turn.

\begin{table}%[!htbp]
\caption{Experimental results on synthetic graph datasets.}
\label{tab:graph_synthetic}
\begin{threeparttable}
\begin{tabular}{ccccc}
\toprule
\multirow{2}{*}{} & \multicolumn{2}{c}{RG-Macro} & \multicolumn{2}{c}{RG-Sub} \\
 & Acc & F1\_ma & Acc & F1\_ma \\
 \midrule
Graphlet Kernel & 0.3800 & 0.3195 & 0.6020 & 0.5919 \\
WL Subtree & \textbf{0.6860} & \textbf{0.6702} & 0.9190 & 0.9218 \\
Hypergraph Rooted & 0.1790 & 0.0933 & 0.1850 & 0.0866 \\
Hypergraph Directed Line & 0.6040 & 0.5441 & 0.6820 & 0.6851 \\
\midrule
Hypergraph WL Subtree & \textbf{0.6860} & \textbf{0.6702} & 0.9190 & 0.9218 \\
Hypergraph WL Hyperedge & 0.6490 & 0.6347 & \textbf{0.9250} & \textbf{0.9268} \\
\bottomrule
\end{tabular}
    \begin{tablenotes}
        \item The best results are marked in bold type.
    \end{tablenotes}
\end{threeparttable}
\end{table}

\begin{table*}%[!htbp]
\centering
\caption{Experimental results on real-world graph datasets.}
\label{tab:graph_real}
\begin{threeparttable}
\begin{tabular}{ccccccccccc}
\toprule
\multirow{2}{*}{} & \multicolumn{2}{c}{IMDB-BINARY} & \multicolumn{2}{c}{IMDB-MULTI} & \multicolumn{2}{c}{MUTAG} & \multicolumn{2}{c}{NCI1} & \multicolumn{2}{c}{PROTEINS} \\
 & Acc & F1\_ma & Acc & F1\_ma & Acc & F1\_ma & Acc & F1\_ma & Acc & F1\_ma \\
 \midrule
Graphlet & 0.6010 & 0.5959 & 0.3700 & 0.3402 & 0.8087 & 0.7599 & 0.6015 & 0.6003 & 0.6864 & 0.6236 \\
Graph WL Subtree & \textbf{0.7440} & \textbf{0.7245} & \textbf{0.5167} & \textbf{0.5053} & \textbf{0.8245} & \textbf{0.7822} & \textbf{0.7487} & \textbf{0.7486} & \textbf{0.7196} & \textbf{0.6950} \\
Hypergraph Rooted & 0.4950 & 0.4201 & 0.3447 & 0.3244 & 0.6649 & 0.3983 & 0.4947 & 0.4715 & 0.5957 & 0.4451 \\
Hypergraph Directed Line & 0.6370 & 0.6361 & 0.4367 & 0.4161 & 0.6650 & 0.3986 & 0.6591 & 0.6560 & 0.6738 & 0.6493 \\
\midrule
Hypergraph WL Subtree & \textbf{0.7440} & \textbf{0.7245} & \textbf{0.5167} & \textbf{0.5065} & \textbf{0.8245} & \textbf{0.7822} & \textbf{0.7487} & \textbf{0.7486} & \textbf{0.7196} & \textbf{0.6950} \\
Hypergraph WL Hyperedge & 0.7250 & 0.7232 & 0.5153 & 0.5049 & 0.8141 & 0.7738 & 0.7037 & 0.7037 & 0.7071 & 0.6760 \\
\bottomrule
\end{tabular}
    \begin{tablenotes}
        \item The best results are marked in bold type.
    \end{tablenotes}
\end{threeparttable}
\end{table*}

\subsubsection{Experimental Results and Discussions}
Experimental results on seven graph datasets are provided in Table \ref{tab:graph_synthetic} and \ref{tab:graph_real}. We have the following three observations. First, the results show that the hypergraph WL subtree performs similarly to the graph WL Subtree method with the same random seeds and fold splitting. It can be explained by the analysis in Section \ref{sec:relation_graph_wl_and_hypergraph_wl}, which indicates that the hypergraph WL kernel equals the graph WL kernel confronting the low-order graph structure. Second, the two compared hypergraph kernel methods perform worse than the graph kernel methods. This is because the two hypergraph kernel methods only focus on the high-order correlations and are ineffective for the common low-order correlations. For example, the hypergraph directed line kernel embeds the graph/hypergraph by two-step structure transformation, which will produce many redundancy low-order structures. Confronting low-order structures tasks like graph classification, those redundant structures include much noise, which may reduce the quality of the final graph/hypergraph features. Those. Third, the hypergraph WL hyperedge kernel performs passably in those graph datasets compared with the hypergraph/graph WL subtree kernel. The main reason is that the edge in graphs only connects two vertices, which restricts the modeling ability of the hyperedge kernel. In the following hypergraph experiments, we will show the significant performance improvement of the proposed two hypergraph kernel compared with existing graph and hypergraph kernel methods.

\subsection{Experiments on Hypergraph Datasets}
\label{sec:exp:hypergraph}
This subsection presents the experiments results on a collection of hypergraph datasets, including four synthetic and eight real-world hypergraph datasets. Table \ref{tab:stat_hypergraph} provides a summary of the dataset statistics..

\begin{table*}%[!htbp]
\centering
\caption{Statistic information of ten hypergraph datasets.}
\label{tab:stat_hypergraph}
\begin{threeparttable}
\begin{tabular}{ccccccc}
\toprule
Dataset & \# hypergraphs & \# classes & Multi-label & Avg. \# vertices & Avg. \# hyperedges & Avg. hyperedge degree \\
\midrule
RHG-10 & 2000 & 10 & \XSolidBrush & 31.3 & 29.8 & 5.2 \\
RHG-3 & 1500 & 3 & \XSolidBrush & 35.5 & 17.9 & 6.9 \\
RHG-Table & 1000 & 2 & \XSolidBrush & 36.3 & 63.3 & 3.1 \\
RHG-Pyramid & 1000 & 2 & \XSolidBrush & 28.8 & 30.6 & 3.0 \\
\midrule
IMDB-Dir-Form & 1869 & 3 & \XSolidBrush & 15.7 & 39.2 & 3.7 \\
IMDB-Dir-Genre & 3393 & 3 & \XSolidBrush & 17.3 & 36.4 & 3.8 \\
IMDB-Dir-Genre-M & 1554 & 6 & \Checkmark & 15.7 & 40.8 & 3.8 \\
IMDB-Wri-Form & 374 & 4 & \XSolidBrush & 10.1 & 3.7 & 5.0 \\
IMDB-Wri-Genre & 1172 & 6 & \XSolidBrush & 12.8 & 4.4 & 5.2 \\
IMDB-Wri-Genre-M & 344 & 7 & \Checkmark & 10.3 & 3.7 & 5.0 \\
Steam-Player & 2048 & 2 & \XSolidBrush & 13.8 & 46.4 & 4.5 \\
Twitter-Friend & 1310 & 2 & \XSolidBrush & 21.6 & 84.3 & 4.3 \\
\bottomrule
\end{tabular}
    \begin{tablenotes}
        \item ``\#'' denotes ``number of'', and ``Avg.'' is short for ``Average''.
    \end{tablenotes}
\end{threeparttable}
\end{table*}

\subsubsection{Datasets}

\textit{Synthetic Hypergraph Datasets:} To evaluate the performance of isomorphism algorithms in capturing structural information, we created four sets of randomly generated hypergraph datasets: RHG-10, RHG-3, RHG-Table, and RHG-pyramid. By randomly picking constructor parameters, we created a range of diverse hypergraphs for experimentation. The RHG-10 dataset encompasses all ten distinct hypergraph structure factors. Table \ref{tab:hypergraph_syn}
showcases examples of all ten typical hypergraph structures. Among the ten manually selected structures, there are three pairs (six types) of hypergraph structures that are prone to confusion, while the remaining four types are selected from common and classical hypergraph structures. To evaluate the algorithm's ability to recognize significant high-order structures, we randomly generated 1000 hypergraphs for three distinctively various hypergraph structures: hyper-pyramid, hyper-check-table, and hyper-wheel, thus constructing the RHG-3 dataset. Moreover, to validate the fine-grained classification capability of our hypergraph isomorphism algorithm, we created two separate datasets for each pair of easily confusable hypergraph structures: RHG-Pyramid comprises one thousand hypergraphs consisting of Hyper-pyramid and Hyper-firm-pyramid structures; RHG-Table consists of data from Hyper-check-table and Hyper-rot-check-table structures.% More details and analysis are provided in Appendix xxx.

\begin{table*}%[!htbp]
\centering
    \caption{Examples of synthetic factor hypergraphs in RHG datasets.}
    \label{tab:RHG_example}
    \scriptsize
    % \footnotesize
  \begin{tabular}{ |c|c|c|c|c| }
    \hline
        \centering
        Hyper Flower&Hyper Pyramid&Hyper Checked Table& Hyper Wheel&Hyper Lattice\\
        \hline
    \begin{minipage}[b]{0.3\columnwidth}
		\centering
		\raisebox{-.5\height}{\includegraphics[width=\linewidth]{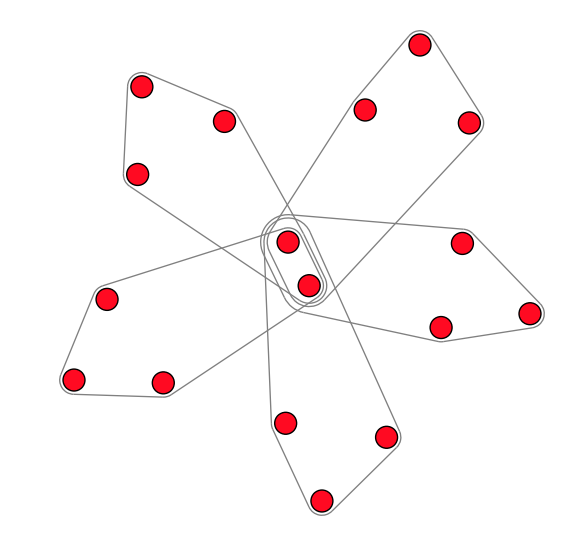}}
	\end{minipage}
    &
    \begin{minipage}[b]{0.3\columnwidth}
		\centering
		\raisebox{-.5\height}{\includegraphics[width=\linewidth]{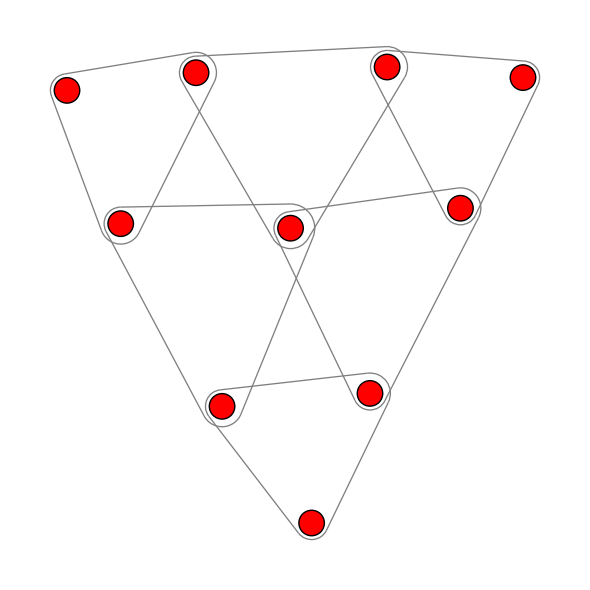}}
	\end{minipage}
    &
    \begin{minipage}[b]{0.3\columnwidth}
		\centering
		\raisebox{-.5\height}{\includegraphics[width=\linewidth]{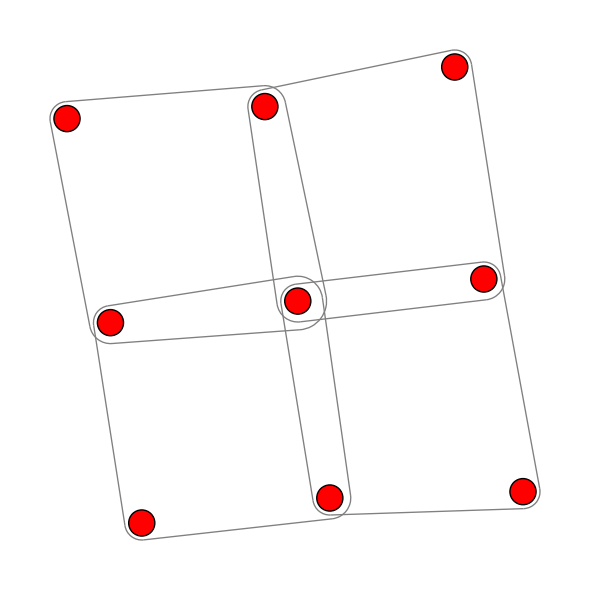}}
	\end{minipage}
    &
    \begin{minipage}[b]{0.3\columnwidth}
		\centering
		\raisebox{-.5\height}{\includegraphics[width=\linewidth]{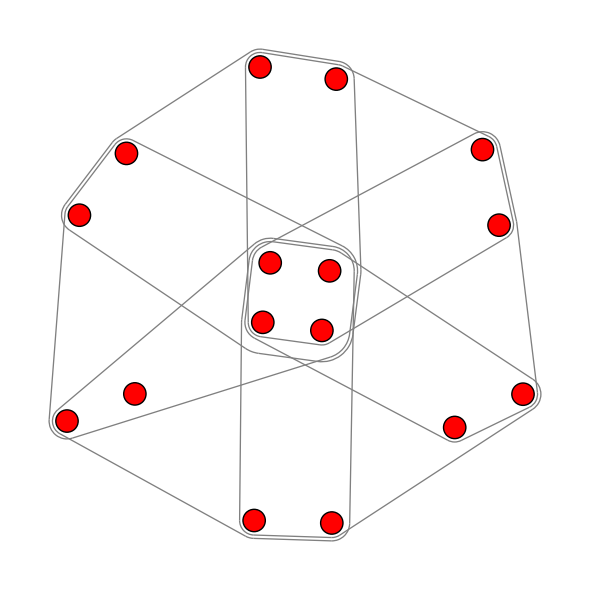}}
	\end{minipage}
    &
    \begin{minipage}[b]{0.3\columnwidth}
		\centering
		\raisebox{-.5\height}{\includegraphics[width=\linewidth]{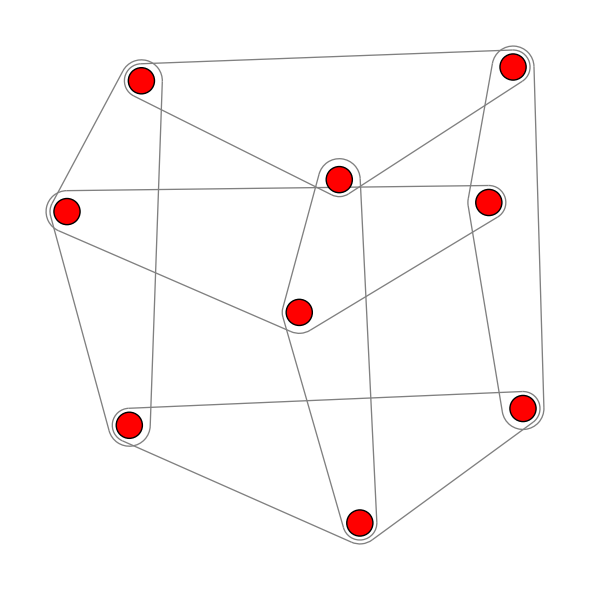}}
	\end{minipage}
    \cr
    \hline
        \centering
        Hyper Windmill&Hyper Firm Pyramid&Hyper Rot Checked Table&Hyper Cycle&Hyper Fern\\
        \hline
    \begin{minipage}[b]{0.3\columnwidth}
		\centering
		\raisebox{-.5\height}{\includegraphics[width=\linewidth]{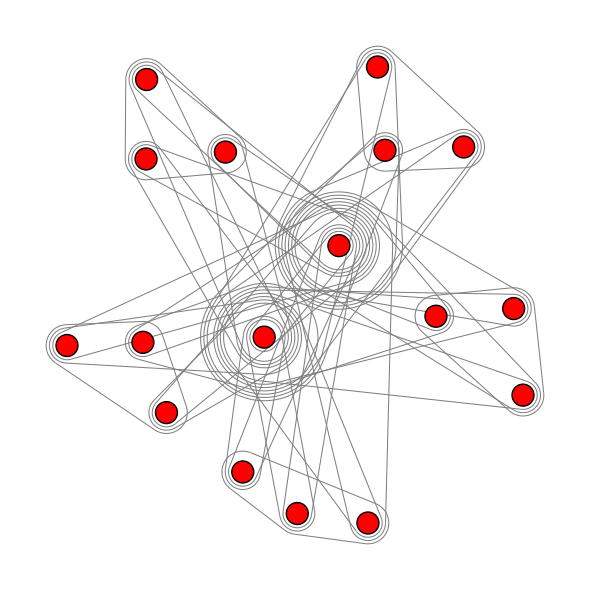}}
	\end{minipage}
    &
    \begin{minipage}[b]{0.3\columnwidth}
		\centering
		\raisebox{-.5\height}{\includegraphics[width=\linewidth]{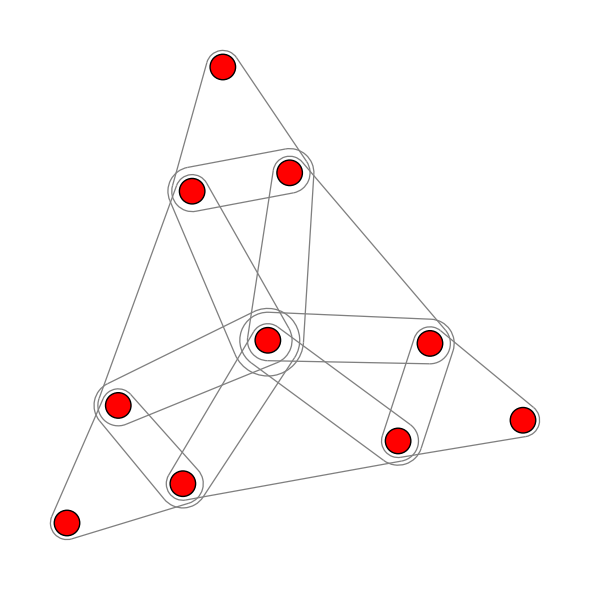}}
	\end{minipage}
    &
    \begin{minipage}[b]{0.3\columnwidth}
		\centering
		\raisebox{-.5\height}{\includegraphics[width=\linewidth]{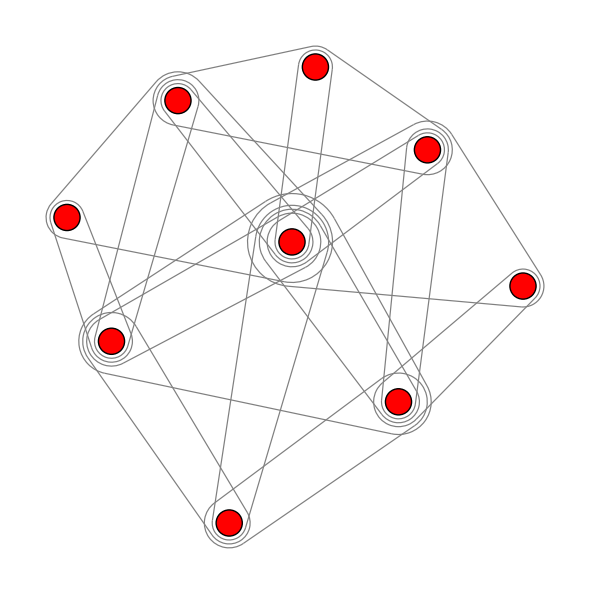}}
	\end{minipage}
    &
    \begin{minipage}[b]{0.3\columnwidth}
		\centering
		\raisebox{-.5\height}{\includegraphics[width=\linewidth]{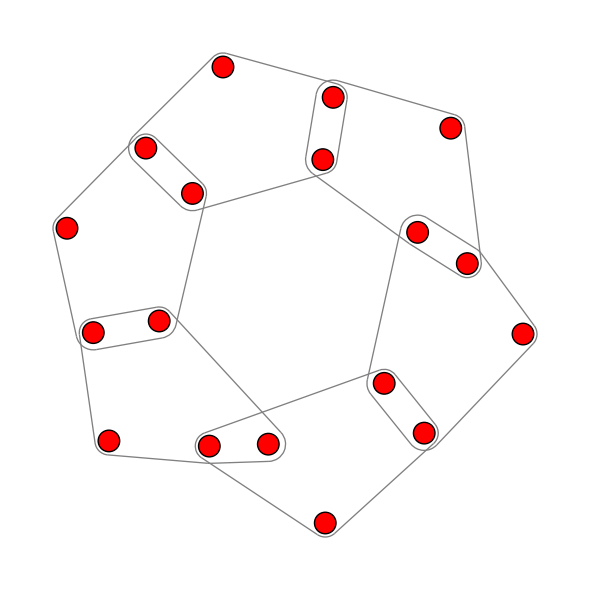}}
	\end{minipage}
    &
    \begin{minipage}[b]{0.3\columnwidth}
		\centering
		\raisebox{-.5\height}{\includegraphics[width=\linewidth]{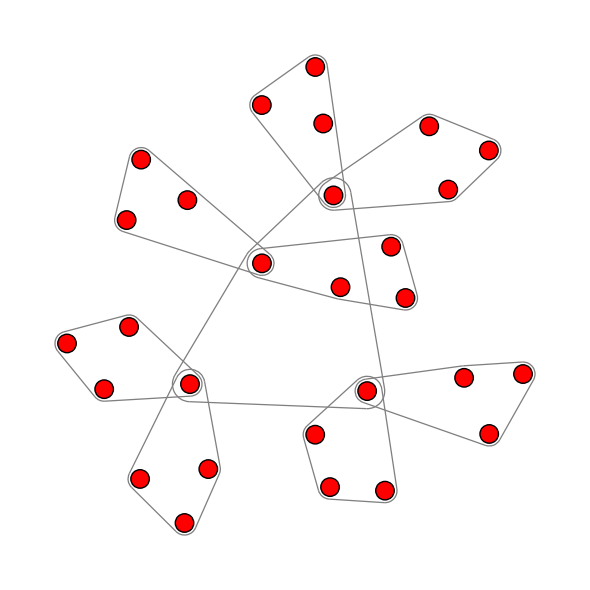}}
	\end{minipage}
    \\ \hline
  \end{tabular}
\end{table*}

\textit{Real-world Hypergraph Datasets:} In Table \ref{tab:stat_hypergraph}, those datasets that starts with ``IMDB'' is constructed from the original IMDB dataset \footnote{\href{https://developer.imdb.com/non-commercial-datasets/}{https://www.imdb.com/}}. Here, two types of correlations: co-director and co-writer, are adopted for high-order dataset construction. In the dataset name, the ``Dir'' and ``Wri'' denotes that the hypergraph is constructed by the co-director and co-writer relationship, respectively. The staff (director/writer) of each movie is a hypergraph. ``Form'' is included in the dataset's name, indicating that the movie category is identified by its form, like animation, documentary, and drama. ``Genre'' denotes the movie is classified by its genres, like adventure, crime, and family. Suffix ``-M'' in the dataset name denotes that each movie is associated with multiple genre labels, which are collected for the more complex multi-label classification task. The Steam-Player dataset \footnote{\href{https://www.kaggle.com/datasets/antonkozyriev/game-recommendations-on-steam}{https://store.steampowered.com/}} is a player dataset where each player is a hypergraph. The vertex is the games played by the player, and the hyperedge is constructed by linking the games with shared tags. The target of the dataset is to identify each user's preference: single-player-game or multi-player-game. The Twitter-Friend \footnote{\href{https://www.kaggle.com/datasets/hwassner/TwitterFriends}{https://twitter.com/}} dataset is a social media dataset. Each hypergraph is the friends of a specified user. The hyperedge is constructed by linking the users who are friends. The label associated with the hypergraph is to identify whether the user posted the blog about ``national dog day'' or ``respect tyler joseph''.

\subsubsection{Results on Synthetic Hypergraphs and Discussions}
Experimental results on four synthetic hypergraph datasets are provided in Table \ref{tab:hypergraph_syn}. Based on the obtained results, distinct observations were made for each dataset. First, the results obtained from the RHG-10 datasets demonstrated significant improvements in our isomorphism algorithms compared to others. Specifically, there was a $ 25.1\% $ accuracy improvement on Hypergraph WL Subtree and a $ 20.3\% $ accuracy improvement on Hypergraph WL Subtree. Second, the isomorphism results provided by Graphlet exhibited considerably lower accuracy than other algorithms. This discrepancy can be attributed to the heavy reliance of Graphlet on random sampling, which compromises the stability of predictions.

\begin{figure}[!htbp]
\centering
    \includegraphics[width=0.45\textwidth]{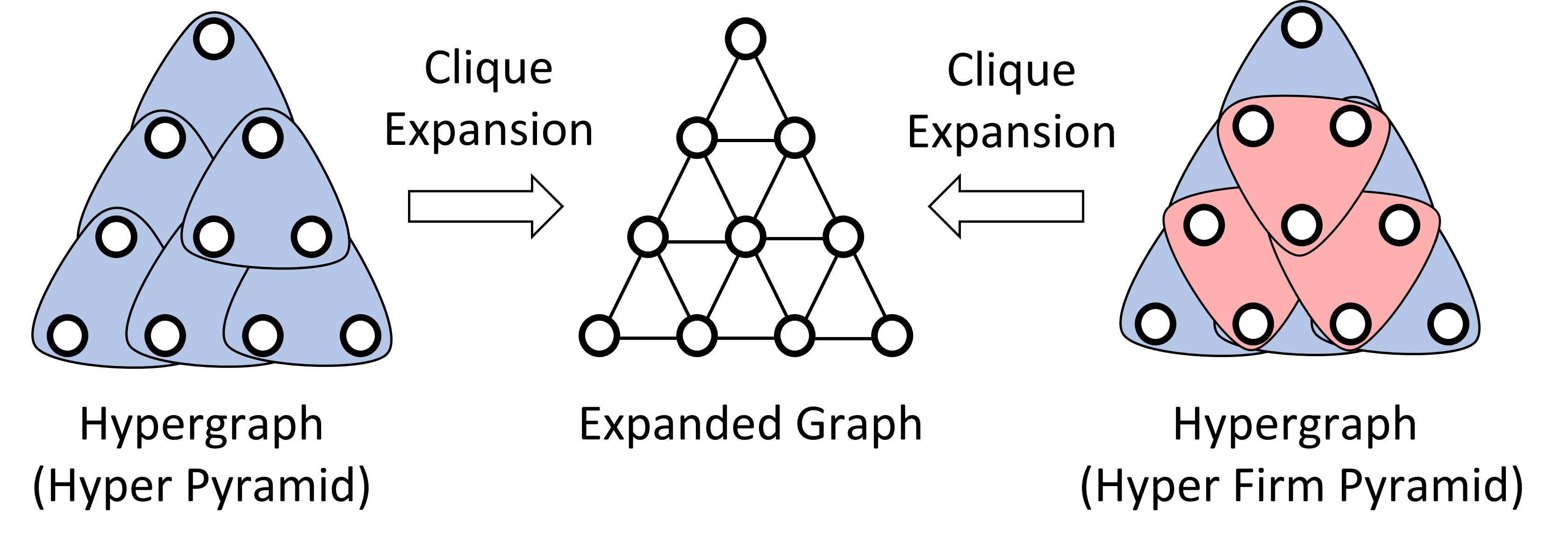}
    \caption{Illustration of the difference between Hyper Firm Pyramid and Hyper Pyramid structures. The difference is marked in red color.}
    \label{fig:pyramid}
\end{figure}

The performance disparity between hypergraph isomorphism and graph isomorphism was evident in the results obtained from the RHG-Table dataset. The need to transform the hypergraph structure into a graph structure through clique expansion when employing the graph isomorphism algorithm revealed that graph structures are less capable of representing complex structures than hypergraphs. Similarly, in the case of RHG-Table and RHG-Pyramid datasets, which consisted of two intricate hypergraph structures, the Hypergraph Rooted algorithm surprisingly exhibited significantly lower recognition ability than other hypergraph isomorphism algorithms. Resembling the performance of the two graph isomorphism algorithms, the Hypergraph Rooted algorithm's accuracy of approximately $ 50\% $ indicated its inability to resolve Hyper-firm-check-table and Hyper-check-table structures. This inferior performance of the Hypergraph Rooted algorithm can be attributed to its fundamental nature as a random walk sampling method, which leads to random traversal when encountering overlapping parts of hyperedges. Consequently, the algorithm lacks a robust judgment of overlapping node groups, common in Hyper Firm Pyramid structures as Figure \ref{fig:pyramid} shows. However, the HG Directed Line algorithm does not suffer from this limitation. The overlapping edges are appropriately reinforced by converting it into a bipartite graph, thereby overcoming the aforementioned issue. Our two algorithms go beyond conventional approaches. They represent a significant advancement by substantially reducing the time needed for extracting structural features from special structures by eliminating the hypergraph-to-graph conversion step and demonstrating a robust recognition ability, particularly in scenarios involving prevalent overlapping hyperedges.

\subsubsection{Results on Real-World Hypergraphs and Discussions}
% 1. 我们方法都高
% 2. 我们方法的两个在不同的数据集高的不一样
% 3. 对比方法分成两类，其中基于结构count的和采样。这两类分别比，超图均高于图。
The Experimental results for eight real-world hypergraph datasets are presented in Tables \ref{tab:hypergraph_real} and \ref{tab:hypergraph_real_m}. Based on those results, we have three observations. Firstly, our proposed two hypergraph WL kernels outperform other compared methods across all real-world hypergraph datasets. This superiority stems from our methods' ability to effectively capture and identify various high-order structures through the two-stage hypergraph WL algorithm. Secondly, we observed that the proposed "HG WL Hyperedge" method outperforms the proposed "HG WL Subtree" method in three IMDB-Wri dataset: IMDB-Wri-Form, IMDB-Wri-Genre, and IMDB-Wri-Genre-M. With further analysis, in Table \ref{tab:stat_hypergraph}, we find that these three hypergraph datasets have lower average numbers of hyperedge ($<= 5.0$) and higher average hyperedge degree ($>=5.0$) compared with other hypergraph datasets. The results indicate that, in comparison to the "HG WL Subtree" method, the "HG WL Hyperedge" method is capable of identifying more complex hyperedges characterized by connecting a greater number of vertices. Thus, the "HG WL Hyperedge" performs better in the three hypergraph datasets. Thirdly, in terms of the compared methods, in most datasets, we observed that, in most datasets, the HG Directed Line method outperforms the Graph WL Subtree method, while the HG Rooted method outperforms the Graphlet method. It is important to note that, to ensure a fair comparison, the Graphlet and HG Rooted methods are sampling-based methods, whereas the other two methods are not. It is evident that the sampling-based methods yield lower performance compared to the other two methods due to the inherent information loss associated with sampling. Within each type of method, the hypergraph kernel approach is able to directly capture and process high-order correlations present in hypergraphs, leading to superior performance.

\begin{table*}%[!htbp]
\centering
\caption{Experimental results of synthetic hypergraph datasets.}
\label{tab:hypergraph_syn}
\begin{threeparttable}
\begin{tabular}{ccccccccc}
\toprule
\multirow{2}{*}{} & \multicolumn{2}{c}{RHG-10} & \multicolumn{2}{c}{RHG-3} & \multicolumn{2}{c}{RHG-Table} & \multicolumn{2}{c}{RHG-Pyramid} \\
 & Acc & F1\_ma & Acc & F1\_ma & Acc & F1\_ma & Acc & F1\_ma \\
 \midrule
Graphlet & 0.3925 & 0.3467 & 0.7033 & 0.7027 & 0.4920 & 0.4627 & 0.4700 & 0.4509 \\
Graph WL Subtree & 0.6720 & 0.6483 & 0.9986 & 0.9986 & 0.5070 & 0.4928 & 0.4940 & 0.4737 \\
Hypergraph Rooted & 0.5545 & 0.5118 & 0.9986 & 0.9986 & \textbf{1.0000} & \textbf{1.0000} & 0.5010 & 0.4686 \\
Hypergraph Directed Line & 0.7040 & 0.6743 & 0.9986 & 0.9986 & \textbf{1.0000} & \textbf{1.0000} & 0.9410 & 0.9406 \\
Hypergraph WL Subtree & \textbf{0.9610} & \textbf{0.9606} & \textbf{0.9993} & \textbf{0.9993} & \textbf{1.0000} & \textbf{1.0000} & \textbf{0.9540} & \textbf{0.9540} \\
Hypergraph WL Hyperedge & 0.9030 & 0.9024 & \textbf{0.9993} & \textbf{0.9993} & \textbf{1.0000} & \textbf{1.0000} & \textbf{0.9540} & \textbf{0.9540} \\
\bottomrule
\end{tabular}
    \begin{tablenotes}
        \item ``HG'' is short for ``Hypergraph''. The best results are marked in bold type.
    \end{tablenotes}
\end{threeparttable}
\end{table*}

\begin{table*}%[!htbp]
\centering
\caption{Experimental results of single-label classification on real-world hypergraph datasets.}
% \scriptsize
\footnotesize
\label{tab:hypergraph_real}
\begin{threeparttable}
\begin{tabular}{ccccccccccccc}
\toprule
\multirow{2}{*}{} & \multicolumn{2}{c}{IMDB-Dir-Form} & \multicolumn{2}{c}{IMDB-Dir-Genre} & \multicolumn{2}{c}{IMDB-Wri-Form} & \multicolumn{2}{c}{IMDB-Wri-Genre} & \multicolumn{2}{c}{Steam-Player} & \multicolumn{2}{c}{Twitter-Friend} \\
 & Acc & F1\_ma & Acc & F1\_ma & Acc & F1\_ma & Acc & F1\_ma & Acc & F1\_ma & Acc & F1\_ma \\
 \midrule
Graphlet & 0.5265 & 0.3922 & 0.6086 & 0.5058 & 0.4654 & 0.2216 & 0.3353 & 0.2318 & 0.4805 & 0.3440 & 0.5878 & 0.3769 \\
Graph WL Subtree & {0.6453} & {0.5967} & 0.7424 & 0.6619 & 0.4681 & 0.2994 & 0.4846 & 0.4202 & 0.5293 & 0.5189 & 0.5931 & 0.4467 \\
HG Rooted & 0.5944 & 0.5247 & 0.6891 & 0.6095 & 0.4438 & 0.2471 & 0.4215 & 0.3755 & 0.4868 & 0.4855 & 0.5817 & 0.4635 \\
HG Directed Line & 0.6670 & 0.6199 & OOM & OOM & 0.4574 & 0.2604 & 0.4999 & 0.4269 & 0.5312 & 0.5206 & OOM & OOM \\
\midrule
HG WL Subtree & \textbf{0.6741} & \textbf{0.6227} & \textbf{0.7804} & \textbf{0.7379} & 0.4679 & \textbf{0.3335} & 0.5307 & 0.4738 & \textbf{0.5620} & \textbf{0.5474} & \textbf{0.5954} & \textbf{0.4800} \\
HG WL Hyperedge & 0.6688 & 0.6197 & {0.7698} & {0.7165} & \textbf{0.4732} & 0.2990 & \textbf{0.5614} & \textbf{0.4953} & 0.5493 & 0.5467 & 0.5809 & 0.3971 \\
\bottomrule
\end{tabular}
    \begin{tablenotes}
        \item ``HG'' is short for ``Hypergraph'', and ``OOM'' denotes ``out of memory''. The best results are marked in bold type.
    \end{tablenotes}
\end{threeparttable}
\end{table*}

\begin{table}%[!htbp]
\caption{Experimental results of multi-label classification on real-world hypergraph datasets.}
\scriptsize
% \footnotesize
\centering
\label{tab:hypergraph_real_m}
\begin{threeparttable}
\begin{tabular}{cp{0.6cm}<{\centering}p{0.6cm}<{\centering}p{0.6cm}<{\centering}p{0.6cm}<{\centering}p{0.6cm}<{\centering}p{0.6cm}<{\centering}}
% \begin{tabular}{ccccccc}
\toprule
\multirow{2}{*}{} & \multicolumn{3}{c}{IMDB-Dir-Genre-M} & \multicolumn{3}{c}{IMDB-Wri-Genre-M} \\
 & EMR & EB\_acc & EB\_pre & EMR & EB\_acc & EB\_pre \\
 \midrule
Graphlet & 0.1795 & 0.4008 & 0.5663 & 0.0524 & 0.3306 & 0.5740 \\
Graph WL Subtree & 0.3533 & 0.4823 & 0.5954 & 0.2122 & 0.4327 & 0.6199 \\
HG Rooted & 0.2799 & 0.4159 & 0.5357 & 0.1626 & 0.3789 & 0.5762 \\
HG Directed Line & 0.3410 & 0.4779 & 0.6075 & 0.1684 & 0.4249 & 0.6131 \\
\midrule
HG WL Subtree & \textbf{0.3739} & \textbf{0.5199} & \textbf{0.6500} & 0.2238 & 0.4463 & \textbf{0.6620} \\
HG WL Hyperedge & 0.3539 & 0.5014 & 0.6210 & \textbf{0.2616} & \textbf{0.4668} & 0.6586\\
\bottomrule
\end{tabular}
    \begin{tablenotes}
        \item ``HG'' is short for ``Hypergraph''. The best results are marked in bold type.
    \end{tablenotes}
\end{threeparttable}
\end{table}

\subsection{Runtime Comparison}

\begin{figure*}[!htbp]
\centering
    \subfigure[]{
    	\begin{minipage}[t]{0.22\textwidth}
    		\centering
    		\includegraphics[width=\textwidth]{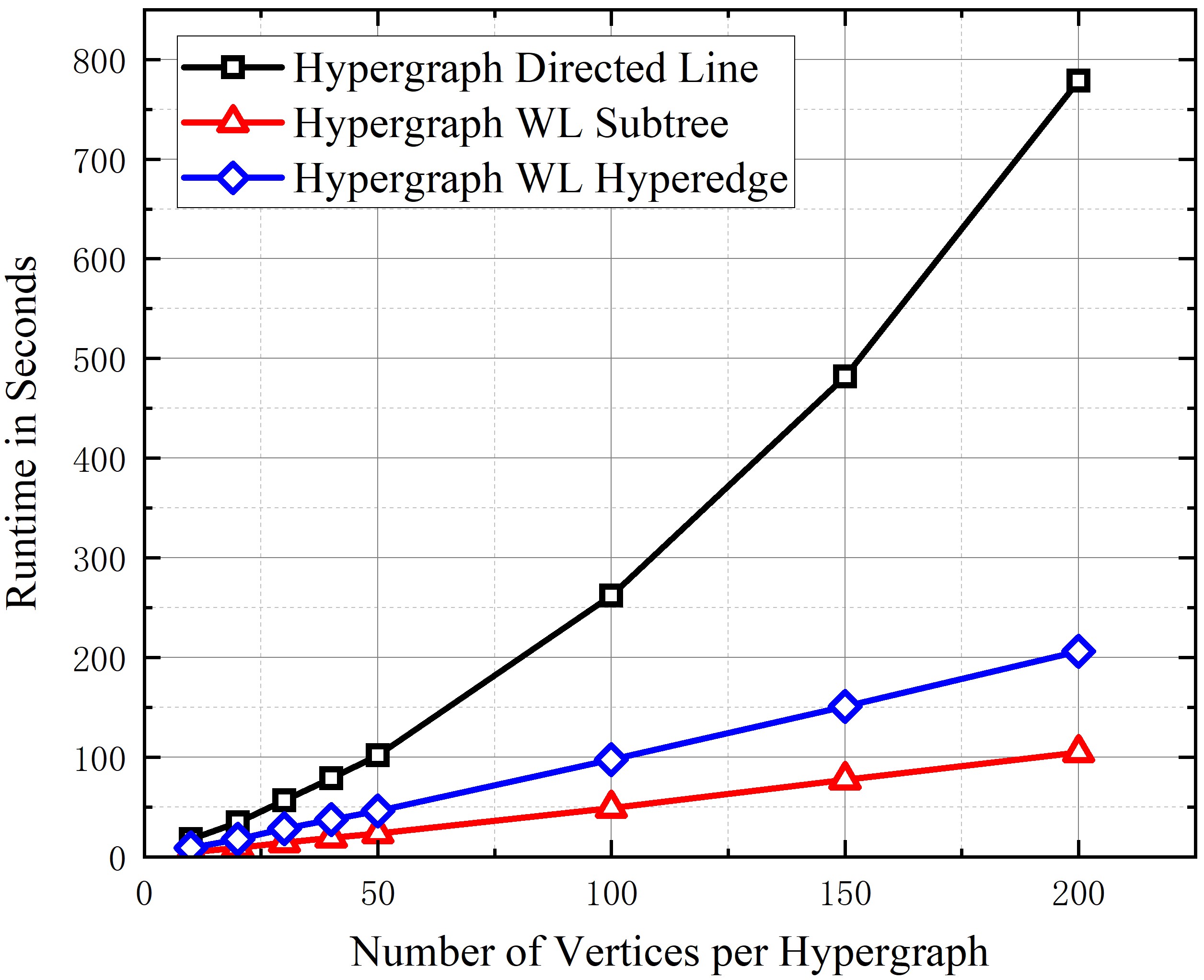}
    		\label{fig:rt_nv}
    	\end{minipage}
	}
    \subfigure[]{
    	\begin{minipage}[t]{0.22\textwidth}
    		\centering
    		\includegraphics[width=\textwidth]{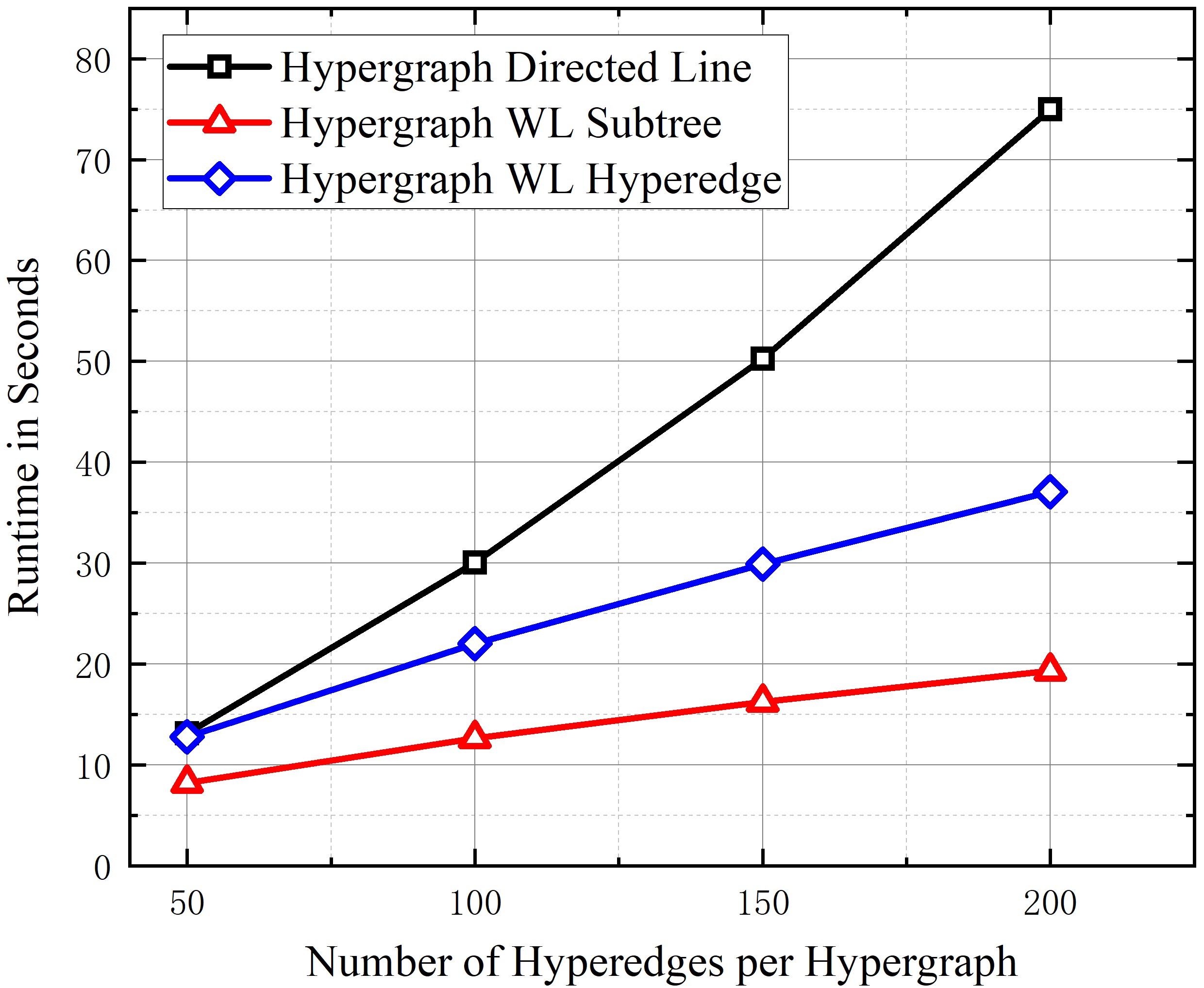}
    		\label{fig:rt_ne}
    	\end{minipage}
	}
    \subfigure[]{
    	\begin{minipage}[t]{0.22\textwidth}
    		\centering
    		\includegraphics[width=\textwidth]{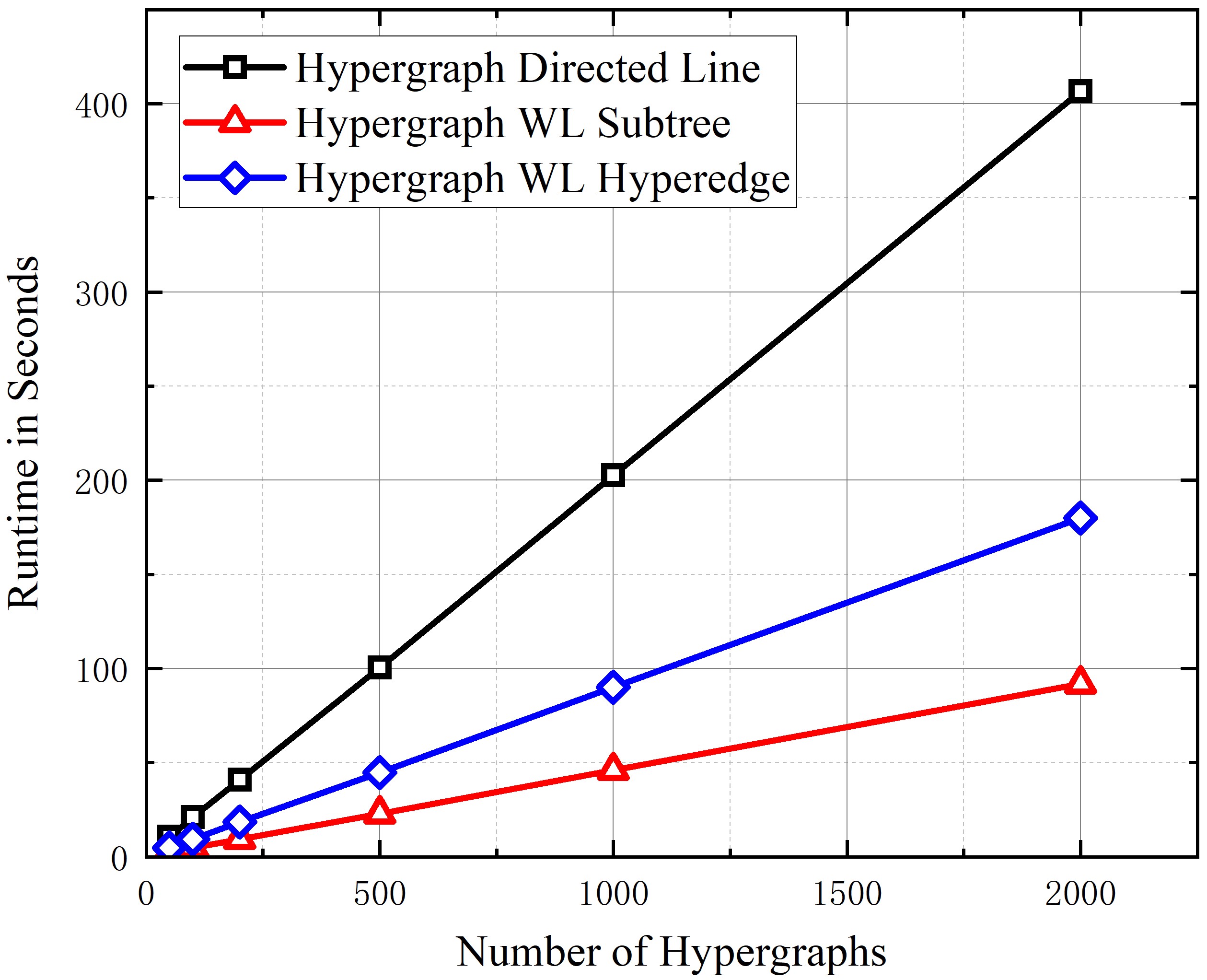}
    		\label{fig:rt_hg}
    	\end{minipage}
	}
    \subfigure[]{
    	\begin{minipage}[t]{0.22\textwidth}
    		\centering
    		\includegraphics[width=\textwidth]{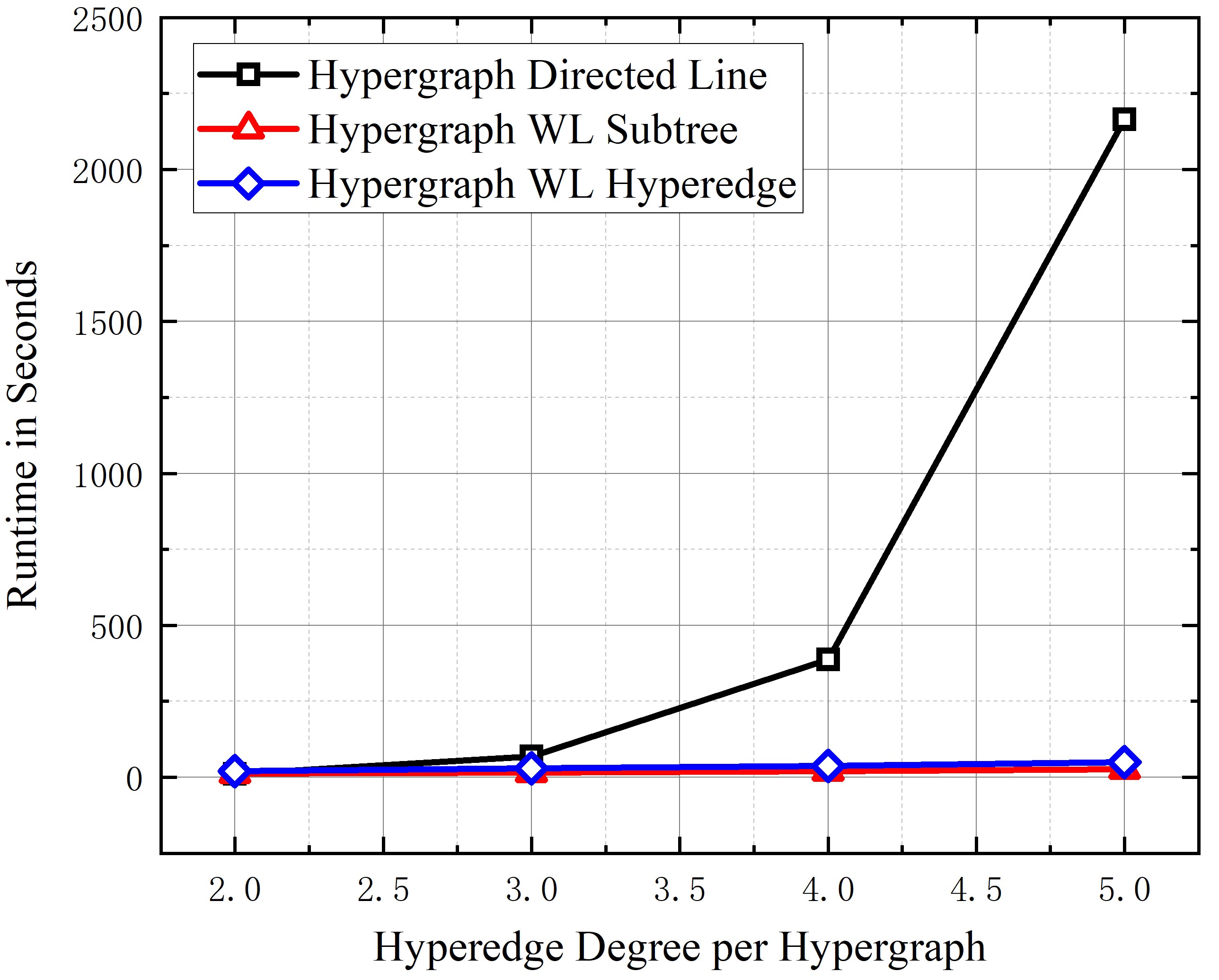}
    		\label{fig:rt_de}
    	\end{minipage}
	}
 
	\caption{
	\label{fig:rt_cmp}
	Runtime comparison for kernel matrix computation on synthetic hypergraphs.}
\end{figure*}

% 超图的个数

% 超图的节点数量

% 超边的个数

% 超图的节点度

In this subsection, we perform experiments to evaluate the runtime performance of different methods, as shown in Figure \ref{fig:rt_cmp}. For a fair comparison, we selected the Hypergraph Directed Line method as the benchmark since it is a hypergraph kernel-based method that does not employ sampling strategies. The DeepHyperGraph\footnote{\url{https://github.com/iMoonLab/DeepHypergraph}} library is utilized to generate the synthetic hypergraphs. Those methods are computed on the same machine with Intel i7-10700 @ 2.90GHz$\times$16 CPU and 16G Memory. By default, $500$ hypergraphs with the ``low-order first'' configuration \footnote{\href{https://deephypergraph.readthedocs.io/en/latest/api/random.html\#dhg.random.hypergraph\_Gnm}{hypergraph\_Gnm()}} are randomly generated, and each hypergraph contains $50$ vertices and $250$ hyperedges. For complete runtime comparison, we devise four types of settings of synthetic hypergraphs. The first synthetic hypergraph setting is developed to test the runtime of methods as the number of vertices per hypergraph grows. In this setting, the number of hypergraphs is fixed at $500$. In each hypergraph, the number of hyperedges is five times the number of vertices. The Experimental results for this setting are presented in Figure \ref{fig:rt_nv}. The second synthetic hypergraph setting is developed to test the runtime of methods as the number of hyperedges per hypergraph grows. In this setting, the number of vertices is fixed as $50$. The number of hyperedges varies in the range $[50, 100, 150, 200]$. Experimental results of this setting are shown in Figure \ref{fig:rt_ne}. The third synthetic hypergraph setting is designed to assess the runtime performance of methods as the number of hypergraphs rises. In this setting, the number of vertices and hyperedge is fixed as $50$ and $250$, respectively. The number of hypergraphs varies in the range $[50, 100, 200, 500, 1000, 2000]$. Experimental results of this setting are shown in Figure \ref{fig:rt_hg}. The last synthetic hypergraph setting is developed to test the runtime of methods as the complexity of the hypergraph rises. Compared with an edge in simple graphs, the hyperedge can connect more than two vertices. The complexity of hypergraphs will rise as the degree of hyperedge rises. In this setting, the number of vertices and hyperedges is also fixed as $50$ and $250$, respectively. However, the degree of hyperedges varies. In other words, we randomly generate $500$ $2$-uniform hypergraph, $3$-uniform hypergraphs, $4$-uniform hypergraphs, and $5$-uniform hypergraphs for runtime comparison. Experimental results of this setting are shown in Figure \ref{fig:rt_de}. From the four sets of experimental results, we observed that the proposed two methods run significantly faster ($86\times$) than the compared Hypergraph Directed Line methods. Especially confronting more complex hypergraph datasets (the degree of hyperedge rises), the proposed methods still show robust runtime as shown in Figure \ref{fig:rt_de}. This is because the Hypergraph Directed Line method can not directly process the hypergraph, which transforms hypergraphs into undirected graphs with clique expansion and further generate directed line graphs for kernel feature computation. As the degree of hyperedge rises, the scale of the clique-expanded graphs will sharply increase. In contrast, our proposed methods could be applied in hypergraphs without extra transformation, which is not sensitive to the variation of hyperedge's degree. Besides, we find that the runtime of the proposed Hypergraph WL Hyperedge method is slower than the proposed Hypergraph WL Subtree method. The main reason is that in each iteration, the Hypergraph WL Subtree method directly counts the vertex label for the final feature vector. However, the Hypergraph WL Hyperedge method needs to build hyperedge label from those vertex labels, which will consume extra time for computation.

%% file: 6_conclusion.tex
\section{Conclusion}

In this paper, we first present a generalization of the Weisfeiler-Lehman test algorithm from graphs to hypergraphs. To handle the complex neighbor relationships inherent in hypergraphs, we introduce a two-stage strategy that considers the vertex neighbors of a hyperedge and the hyperedge neighbors of a vertex. Based on this approach, we propose the Hypergraph Weisfeiler-Lehman test algorithm as a solution to the hypergraph isomorphism problem.
Secondly, we propose a general hypergraph Weisfeiler-Lehman kernel framework and implement two instances: Hypergraph Weisfeiler-Lehman Subtree Kernel and Hypergraph Weisfiler-Lehman Hyperedge Kernel. We also provide theoretical evidence that, confronting the graph structure, when applied to graph structures, the proposed Hypergraph Weisfeiler-Lehman Kernel exhibits the same power as the conventional Graph Weisfeiler-Lehman Kernel. 
Lastly, we conduct experiments on seven graph classification datasets and twelve hypergraph classification datasets. The results showcase significant performance improvements, affirming the effectiveness of the proposed methods.

%% file: 7_app.tex
% \appendices

% % you can choose not to have a title for an appendix
% % if you want by leaving the argument blank
% \section{}

% % use section* for acknowledgment
% \ifCLASSOPTIONcompsoc
%   % The Computer Society usually uses the plural form
%   \section*{Acknowledgments}
% \else
%   % regular IEEE prefers the singular form
% \fi


%% file: 8_bib.tex
\ifCLASSOPTIONcaptionsoff
  \newpage
\fi